\documentclass[12pt]{article}

\RequirePackage[OT1]{fontenc}
\usepackage{natbib}
\RequirePackage[colorlinks,citecolor=blue,urlcolor=blue]{hyperref}
\usepackage[english]{babel}
\usepackage{amsthm}
\usepackage{amsmath}
\usepackage{amsfonts}
\usepackage{amssymb}
\usepackage{graphicx}
\usepackage{enumerate}
\usepackage{color}
\usepackage{array}
\usepackage{bm}
\usepackage{multirow}
\usepackage{booktabs} 
\usepackage{adjustbox}
\usepackage{caption}
\usepackage{subcaption}
\usepackage{sectsty}
\usepackage{longtable}
\usepackage[figuresright]{rotating}
\usepackage{multicol}
\usepackage{wrapfig}

\def\frac#1#2{{\textstyle{#1\over#2}}}

\DeclareSymbolFont{AMSb}{U}{msb}{m}{n}
\DeclareMathSymbol{\Natural}{\mathbin}{AMSb}{"4E}
\DeclareMathSymbol{\Integer}{\mathbin}{AMSb}{"5A}
\DeclareMathSymbol{\Real}{\mathbin}{AMSb}{"52}
\DeclareMathSymbol{\Rational}{\mathbin}{AMSb}{"51}
\DeclareMathSymbol{\Imaginary}{\mathbin}{AMSb}{"49}
\DeclareMathSymbol{\Complex}{\mathbin}{AMSb}{"43} 
\DeclareMathSymbol{\Disk}{\mathbin}{AMSb}{"44} 
\def\bi{\begin{itemize}}
\def\ei{\end{itemize}}
\def\bd{\begin{description}}
\def\ed{\end{description}}
\def\ben{\begin{enumerate}}
\def\een{\end{enumerate}}
\def\bv{\begin{verbatim}}
\def\ev{\end{verbatim}}

\def\calT{{\mathcal T}}

\def\calV{{\mathcal V}}

\def\pr{{\rm Pr}}

\def\E{{\rm E}}

\def\var{{\rm var}}

\def\cov{{\rm cov}}  
\def\corr{{\rm corr}}

\def\Dto{{\ {\buildrel D\over \longrightarrow}\ }}

\def\2to{{\ {\buildrel 2\over \longrightarrow}\ }}

\def\I1ton{{$I_1,\ldots,I_n$}}
\def\X1ton{{$X_1,\ldots,X_n$}}
\def\Y1ton{{$Y_1,\ldots,Y_n$}}
\def\Z1ton{{$Z_1,\ldots,Z_n$}}
\def\R1ton{{$R_1,\ldots,R_n$}}
\def\e1ton{{$e_1,\ldots,e_n$}}
\def\t1ton{{$t_1,\ldots,t_n$}}
\def\x1ton{{$x_1,\ldots,x_n$}}
\def\y1ton{{$y_1,\ldots,y_n$}}
\def\z1ton{{$z_1,\ldots,z_n$}}

\newtheorem{defn}{Definition}

\newtheorem{prop}[defn]{Proposition}

\begin{document}
\thispagestyle{empty}
\baselineskip=28pt
\vskip 5mm
\begin{center} 
{\Large{\bf Vecchia Likelihood Approximation for Accurate and Fast Inference in Intractable Spatial Extremes Models}}
\end{center}

\baselineskip=12pt
\vskip 5mm

\begin{center}
\large
Rapha\"el Huser$^1$, Michael L. Stein$^2$ and Peng Zhong$^1$
\end{center}

\footnotetext[1]{
\baselineskip=10pt Statistics Program, Computer, Electrical and Mathematical Sciences and Engineering (CEMSE) Division, King Abdullah University of Science and Technology (KAUST), Thuwal 23955-6900, Saudi Arabia. E-mails: raphael.huser@kaust.edu.sa; peng.zhong@kaust.edu.sa}
\footnotetext[2]{
\baselineskip=10pt Department of Statistics, Rutgers University, Piscataway, NJ 08854 United States of America. E-mail: ms2870@stat.rutgers.edu}

\baselineskip=17pt
\vskip 4mm
\centerline{\today}
\vskip 6mm

\begin{center}
{\large{\bf Abstract}}
\end{center}
Max-stable processes are the most popular models for high-impact spatial extreme events, as they arise as the only possible limits of spatially-indexed block maxima. However, likelihood inference for such models suffers severely from the curse of dimensionality, since the likelihood function involves a combinatorially exploding number of terms. In this paper, we propose using the Vecchia approximation, which conveniently decomposes the full joint density into a linear number of low-dimensional conditional density terms based on well-chosen conditioning sets designed to improve and accelerate inference in high dimensions. Theoretical asymptotic relative efficiencies in the Gaussian setting and simulation experiments in the max-stable setting show significant efficiency gains and computational savings using the Vecchia likelihood approximation method compared to traditional composite likelihoods. Our application to extreme sea surface temperature data at more than a thousand sites across the entire Red Sea further demonstrates the superiority of the Vecchia likelihood approximation for fitting complex models with intractable likelihoods, delivering significantly better results than traditional composite likelihoods, and accurately capturing the extremal dependence structure at lower computational cost. 
\baselineskip=16pt

\par\vfill\noindent

\pagenumbering{arabic}
\baselineskip=24pt

\newpage


\section{Introduction}\label{sec:Introduction}
Max-stable models have been used extensively for describing the dependence structure in multivariate and spatial extremes \citep{Padoan.etal:2010,Segers:2012,Davis.etal:2013a,deCarvalho.Davison:2014,Huser.Davison:2014a,Huser.Genton:2016}. They are natural models to use since they are characterized by the max-stability property, which arises in limiting joint distributions for block maxima with block size tending to infinity; see the reviews by \citet{Davison.etal:2012}, \citet{Davison.Huser:2015} and \citet{Davison.etal:2019}.

However, likelihood-based inference for high-dimensional max-stable distributions is computationally prohibitive \citep{Padoan.etal:2010,Castruccio.etal:2016}. Although the likelihood function has a known general expression, it involves a combinatorial explosion of terms, which makes it impossible to evaluate it exactly, even in relatively small dimensions. In classical geostatistics, Gaussian graphical models, which are represented in terms of a conditional independence graph, play a key role for modeling big spatial data as they lead to Gaussian Markov random fields \citep{Rue.Held:2005} with a sparse precision (i.e., inverse covariance) matrix, which are directly linked to certain classes of continuous-space Gaussian stochastic partial equation models \citep{Lindgren.etal:2011}. Thanks to the Hammersley--Clifford Theorem, the joint density of graphical models can be decomposed into lower-dimensional densities according to the underlying graph, thus making computations much faster. In the extremes context, recent work has shown how to build graphical models for multivariate extremes based on high threshold exceedances modeled through the multivariate Pareto distribution \citep{Engelke.Hitz:2020,Engelke.Ivanovs:2021}. However, interestingly, it is possible to show that conditional independence in max-stable models with a continuous joint density already yields full independence \citep{Papastathopoulos.Strokorb:2016}. This implies that non-trivial Markov max-stable models do not exist, and thus, that likelihood-based inference for max-stable processes is not just a challenging task; it is, by nature of the problem, \emph{intrinsically} difficult. In other words, this computational bottleneck is ``built-in'', and cannot be easily bypassed. Nevertheless, viable inference solutions need to be found.
 
For some very specific classes of max-stable models, 
fast methods can still be designed: the likelihood function for the logistic and nested logistic multivariate models can be efficiently computed using a recursive formula (see \citealp{Shi:1995a}, and \citealp{Vettori.etal:2019}), while the hierarchical construction of the Reich--Shaby max-stable spatial process can be exploited to perform Bayesian inference in high dimensions \citep{Reich.Shaby:2012a,Stephenson.etal:2015,Bopp.etal:2021,Vettori.etal:2019}. Apart from these restrictive cases, full likelihood inference for max-stable models is extremely intensive, and this has prevented the use of more flexible max-stable classes, such as the Brown--Resnick \citep{Kabluchko.etal:2009} or extremal-$t$ \citep{Opitz:2013a} processes, in high-dimensional settings. Recent attempts have succeeded in fitting the Brown--Resnick process in dimension $D\approx20$ based on the full likelihood, either using an astute stochastic expectation--maximization algorithm \citep{Huser.etal:2019} or a Markov chain Monte Carlo algorithm in the Bayesian framework \citep{Thibaud.etal:2016,Dombry.etal:2017}. Nevertheless, these approaches remain difficult to apply in higher dimensions. Alternatively, \citet{Stephenson.Tawn:2005} have proposed a full likelihood approach based on the occurrence times of maxima, but \citet{Wadsworth:2015} and \citet{Huser.etal:2016} have found that this is often severely biased in low dependence situations. More recently, \citet{Lenzi.etal:2021} proposed using neural networks for parameter estimation in intractable models, including max-stable processes. They showed that considerable time savings can be obtained, though their machine learning-based approach typically requires the data to be on a regular grid. Moreover, training neural networks for parameter estimation requires model-specific tuning; it becomes very tricky as the number of parameters increases; and, as often the case with machine learning approaches, statistical guarantees are difficult to obtain, especially as far as uncertainty quantification is concerned. 

To make inference for max-stable processes, \citet{Padoan.etal:2010} initially suggested using a pairwise (composite) likelihood, which is built by combining bivariate densities that are possibly weighted to improve statistical and computational efficiency. The benefits of this approach are that (i) it is simple to implement; (ii) it yields dramatic reductions in computational burden with respect to a full likelihood-based approach; and (iii) large-sample properties of composite likelihood estimators are well understood. The main drawback is that it leads to some considerable loss in efficiency due to using only the information contained in pairs of variables. Similarly, a pairwise M-estimator was proposed by \citet{Einmahl.etal:2016}, with optimal, data-driven weights to improve statistical efficiency. In the same spirit, \citet{Padoan.etal:2010} suggested selecting only close-by pairs of sites, i.e., using binary weights set according to the distance between sites, and choosing the cutoff distance in a way that minimizes the trace of the estimator's asymptotic variance. Although this improves the estimator, it is still quite far from optimal, especially in high-dimensional settings. Alternatively, \citet{Genton.etal:2011}, \citet{Huser.Davison:2013a}, \citet{Sang.Genton:2014} and \citet{Castruccio.etal:2016} have explored triplewise and higher-order composite likelihoods, and have shown that significant efficiency gains can be obtained by using \emph{truncated} composite likelihoods, i.e., by choosing the marginal likelihood components that are contained within a disk of fixed radius. However, this approach is still not very attractive in large dimensions $D$, because it is costly to enumerate all the ${D \choose d}$ marginal likelihood components that are built from $2\leq d\leq D$ sites, and to identify and evaluate those that are contained within a disk of radius $\delta>0$. Moreover, unless the truncation distance $\delta$ is very small, the number of such selected components may still be too large to be practical in high dimensions $D$.

In this paper, we propose making inference for max-stable processes by leveraging the Vecchia approximation \citep{Vecchia:1988}. Essentially, the joint density of the data is approximated by a product of well-chosen lower-dimensional conditional densities. Therefore, as explained in Section~\ref{sec:composite}, it can be viewed as a particular type of (weighted) composite likelihood. However, unlike the classical pairwise or higher-order composite likelihood approaches considered previously in the extreme-value literature, the Vecchia approximation provides by construction a \emph{valid} likelihood function, in the sense that it corresponds to the joint density of a well-defined data generating process that approximates the true process under study. Moreover, the number of conditional densities to compute is proportional to the dimension $D$. For these reasons, the Vecchia approximation has been found in the Gaussian-based geostatistical setting not only to provide fast inference for big datasets, but also to generally retain high efficiency compared to full likelihood approaches and to outperform block composite likelihoods \citep{Stein.etal:2004,Katzfuss.etal:2020,Katzfuss.Guinness:2021}. 

The Vecchia approximation relies on the choice of three elements: (i) a permutation defining an ordering of spatial sites; (ii) the number of conditioning sites; and (iii) the conditioning sets themselves. While this flexibility might be seen as a limitation, \citet{Guinness:2018} instead argues that it can be exploited to sharpen the approximation. Based on simulation results, \citet{Guinness:2018} suggested using a maximum-minimum distance ordering, which provides some improvements over coordinate-based orderings. In order to study the exact effect that these three choices have on the Vecchia approximation, and to do a formal comparison with classical composite likelihood approaches, we study in Section~\ref{sec:Gaussian} the theoretical asymptotic relative efficiency of these different estimators in the Gaussian setting for various correlation models. Our new results complement the theoretical results of \citet{Stein.etal:2004} and the numerical results of \citet{Guinness:2018}, \citet{Katzfuss.etal:2020}, and \citet{Katzfuss.Guinness:2021}. In the Supplementary Material, we also study the efficiency gains of an alternative composite likelihood approach that modifies the weights involved in the classical Vecchia approximation. In Section~\ref{sec:MaxStable}, we conduct an extensive simulation study to extend these results to the popular Brown--Resnick max-stable model and, in the Supplementary Material, to the multivariate logistic max-stable model. Our results for the Gaussian and max-stable cases provide evidence that the Vecchia approximation yields competitive efficiency and attractive computational savings, while scaling well with the dimension. We use our results to provide guidance on the choice of the ordering and conditioning sets in the max-stable setting.

In Section~\ref{sec:application}, we exploit the Vecchia approximation to study sea surface temperature extremes for the whole Red Sea at more than a thousand sites. We demonstrate the advantages of using the Vecchia approximation method compared to traditional composite likelihoods. Section~\ref{sec:conclusion} concludes with some discussion and a perspective on future research.

\section{Inference based on composite likelihoods and the Vecchia approximation}\label{sec:composite}

\subsection{Composite likelihoods and choice of weights}\label{subsec:composite}
Consider a $D$-dimensional random vector $\boldsymbol{Z}\in\Real^D$ with density $f(\boldsymbol{z};\boldsymbol{\psi})$, $\boldsymbol{z}\in\mathcal Z\subset\Real^D$, parametrized in terms of a $m$-dimensional vector $\boldsymbol{\psi}=(\psi_1,\ldots,\psi_m)^\top\in\Psi\subseteq\Real^m$. Marginal densities of all subvectors are also denoted by $f$, for simplicity. Suppose that the true parameter vector is $\boldsymbol{\psi}_0\in{\rm Int}(\Psi)$. Then, a composite log-likelihood for $n$ independent realizations $\boldsymbol{z}_1,\ldots,\boldsymbol{z}_n$ of the random vector $\boldsymbol{Z}$ may be defined as $\ell_{C}(\boldsymbol{\psi})=\sum_{i=1}^n \ell_{C}(\boldsymbol{\psi};\boldsymbol{z}_i)$, where
\begin{equation}
\label{eq:composite}
\ell_{C}(\boldsymbol{\psi};\boldsymbol{z})=\sum_{S \in C_{D}}w_{S} \log f(\boldsymbol{z}_S;\boldsymbol{\psi})=\sum_{d=1}^D\sum_{S \in C_{D;d}}w_{S} \log f(\boldsymbol{z}_S;\boldsymbol{\psi}),
\end{equation}
where $C_{D}=\cup_{d=1}^DC_{D;d}$ is the collection of all non-empty subsets of $\{1,\ldots,D\}$, $C_{D;d}$ is the collection of all $d$-dimensional subsets of $\{1,\ldots,D\}$, and $w_S\in\Real$ is a weight attributed to subset $S$. We write $\boldsymbol{z}_S$ to denote the subvectors obtained by restricting $\boldsymbol{z}$ to the components indexed by the subset $S$. The maximum composite likelihood estimator (MCLE) is defined as $\widehat{\boldsymbol{\psi}}_C=\arg\max_{\boldsymbol{\psi}\in\Psi}\ell_C(\boldsymbol{\psi})$. Provided all likelihood terms involved in \eqref{eq:composite} satisfy the Bartlett identities, the gradient of \eqref{eq:composite} with respect to $\boldsymbol{\psi}$ is an unbiased estimating equation, and thus the classical asymptotic theory can be applied. If $\boldsymbol{\psi}$ is identifiable from the likelihood terms with non-zero weight in \eqref{eq:composite}, then under mild regularity conditions, $\widehat{\boldsymbol{\psi}}_C$ is consistent and asymptotically normal as $n\to\infty$ and the variance-covariance matrix of $\widehat{\boldsymbol{\psi}}_C$ can be approximated by $\boldsymbol{V}=n^{-1}\boldsymbol{J}^{-1}(\boldsymbol{\psi}_0)\boldsymbol{K}(\boldsymbol{\psi}_0)\boldsymbol{J}^{-1}(\boldsymbol{\psi}_0)$ for large $n$, where $\boldsymbol{J}(\boldsymbol{\psi})=\E\{-{\partial^2\over\partial\boldsymbol{\psi}\partial\boldsymbol{\psi}} \ell_{C}(\boldsymbol{\psi};\boldsymbol{Z})\}$ is the sensitivity matrix and $\boldsymbol{K}(\boldsymbol{\psi})=\var\{{\partial\over\partial\boldsymbol{\psi}} \ell_{C}(\boldsymbol{\psi};\boldsymbol{Z})\}$ is the variability matrix; see, e.g., \citet{Varin.etal:2011}. The choice of weights $w_S$  in \eqref{eq:composite} turns out to be crucial for the estimator's efficiency. Although weights are often assumed to be non-negative \citep{Varin.etal:2011,Castruccio.etal:2016}, this is non-necessarily restrictive and \citet{Pace.etal:2019} show that optimal weights may in some cases be negative; see also \citet{Fraser.Reid:2019}. As the sum in \eqref{eq:composite} involves $2^D-1$ terms, some weights $w_S$ are usually set to zero for computations.

Composite marginal log-likelihoods of order $d=1,2,\ldots,D$ are defined by setting $w_S=0$ in \eqref{eq:composite} for all subsets $S\subset\{1,\ldots,D\}$ with cardinality $|S|\neq d$. They corresponding composite log-likelihood may be written as $\ell_{C;d}(\boldsymbol{\psi})=\sum_{i=1}^n \ell_{C;d}(\boldsymbol{\psi};\boldsymbol{z}_i)$ with
\begin{equation}\label{eq:composited}
\ell_{C;d}(\boldsymbol{\psi};\boldsymbol{z})=\sum_{S \in C_{D;d}}w_{S} \log f(\boldsymbol{z}_S;\boldsymbol{\psi}).
\end{equation}
We write $\widehat{\boldsymbol{\psi}}_{C;d}$ to denote the mode of $\ell_{C;d}(\boldsymbol{\psi})$. The definition \eqref{eq:composited} includes pairwise ($d=2$) or triplewise ($d=3$) likelihoods that were advocated by \citet{Padoan.etal:2010}, \citet{Genton.etal:2011} and \citet{Huser.Davison:2013a} as a method of inference for max-stable processes, for which the full likelihood is intractable in large dimensions $D$. \citet{Castruccio.etal:2016} also investigated higher-order composite likelihoods of the form \eqref{eq:composited} and reported efficiency gains for increasing $d$. The choice of weights $w_S$ in pairwise likelihoods is not trivial. In the context of max-stable processes, \citet{Padoan.etal:2010} suggest using binary weights $w_S=\mathbb I(\|\boldsymbol{h}_S\|\leq \delta)$, for some cutoff distance $\delta>0$, where $\|\boldsymbol{h}_S\|$ denotes the distance between the pair of sites indexed by the set $S$ and $\mathbb I(\cdot)$ is the indicator function, while they choose $\delta$ by minimizing an estimate of the asymptotic variance $\boldsymbol{V}$. This approach leads to efficiency gains as opposed to using equal weights, i.e., $w_S=1$ for all $S$, but it may not be optimal. \citet{Huser:2013}, Chapter~3, studies the efficiency of pairwise likelihood estimators for Gaussian and max-stable time series models, and provide some further guidance on the choice of weights. For higher-order composite likelihoods with $d>2$, it is even less clear how to select the weights $w_S$ optimally, and by analogy to the pairwise likelihood setting, \citet{Sang.Genton:2014} and \citet{Castruccio.etal:2016} have suggested adopting a \emph{truncated} composite likelihood approach, which uses weights of the type $w_S=\mathbb I(\max_{\{i,j\}\subset S}\|\boldsymbol{h}_{\{i,j\}}\|\leq \delta)$ for some cutoff distance $\delta>0$, thus discarding $d$-dimensional subsets with pairs of sites that are distant from each other.

\subsection{Vecchia approximation}\label{subsec:Vecchia}
The Vecchia approximation \citep{Vecchia:1988} relies on the simple fact that the joint density can be written as the product of conditional densities; see also \citet{Stein.etal:2004}. Consider the vector  $\boldsymbol{z}=(z_1,\ldots,z_D)^\top\in\Real^D$ and a permutation $p:\{1,\ldots,D\}\mapsto\{1,\ldots,D\}$, which defines a re-ordering of the variables $z_j$, $j=1,\ldots,D$. We define the ``history'' of the $j$th variable based on the permutation $p$ as the subvector $\boldsymbol{z}_{H(j;p)}$, where $H(j;p)=\{l\in\{1,\ldots,D\}:p(l)<p(j)\}$ denotes the index set of ``past'' variables. Then, for any choice of permutation $p$, the joint density may be expressed as
\begin{equation}\label{eq:fulldensity}
f(\boldsymbol{z};\boldsymbol{\psi})=f(z_{p(1)};\boldsymbol{\psi})\prod_{j=2}^Df(z_{p(j)}\mid\boldsymbol{z}_{H(j;p)};\boldsymbol{\psi}).
\end{equation}
The Vecchia approximation consists in replacing the history $\boldsymbol{z}_{H(j;p)}$ in \eqref{eq:fulldensity} with a subvector $\boldsymbol{z}_{S(j;p)}$, with $S(j;p)\subseteq H(j;p)$, i.e., 
\begin{equation}\label{eq:densityVecchia}
f_V(\boldsymbol{z};\boldsymbol{\psi}):= f(z_{p(1)};\boldsymbol{\psi})\prod_{j=2}^Df(z_{p(j)}\mid\boldsymbol{z}_{S(j;p)};\boldsymbol{\psi})\approx f(\boldsymbol{z};\boldsymbol{\psi}).
\end{equation}
A counterpart of \eqref{eq:densityVecchia} based on blocks of variables is also considered in \citet{Stein.etal:2004}. While the permutation $p$ is irrelevant for the full density in \eqref{eq:fulldensity}, it affects the approximation \eqref{eq:densityVecchia}. As opposed to time series data, there is no natural ordering of variables in the spatial setting, and although \citet{Stein.etal:2004} argues that it has a negligible impact on the quality of the Vecchia approximation, \citet{Guinness:2018} instead suggests that certain orderings have a better performance than simple coordinate-based orderings. As \citet{Stein.etal:2004} and \citet{Katzfuss.Guinness:2021} show, the Vecchia approximation crucially depends on the size of the conditioning sets $S(j;p)$, which implies is a tradeoff between approximation accuracy and computational efficiency. Usually, a compromise is adopted between singletons of cardinality $|S(j;p)|=1$ (with low computational burden but poor approximation) and maximal sets of cardinality $|S(j;p)|=j-1$ as with the full likelihood (with perfect approximation but heavy computational burden). Here, we choose to restrict the cardinality to $|S(j;p)|=\min(j,d)-1$, for some lower dimension $2\leq d\leq D$. Typically, the ``cutoff dimension'' $d$ will be quite small, which dramatically reduces the computational burden. Finally, the Vecchia approximation \eqref{eq:densityVecchia} also depends on the specific choice of variables to include in the sub-history $S(j;p)$. We here follow the original paper of \citet{Vecchia:1988} who in the spatial context suggest including the $\min(j,d)-1$ nearest neighbors of the $j$-th site among those that belong to its history, $H(j;p)$. Thereafter, we write $S(j;p)\equiv S_{d-1}(j;p)$ to stress that the dimensionality of the conditioning sets is at most $d-1$.

The log-likelihood based on the Vecchia approximation \eqref{eq:densityVecchia} may be written in composite likelihood form as in \eqref{eq:composite}. Precisely, it may be expressed as 
\begin{align}\label{eq:compositeVecchia}
\ell_{V;d}(\boldsymbol{\psi};\boldsymbol{z})&=\log f_{V;d}(\boldsymbol{z};\boldsymbol{\psi}) \nonumber \\
&=\log f(z_{p(1)};\boldsymbol{\psi}) + \sum_{j=2}^D \log f(z_{p(j)},\boldsymbol{z}_{S_{d-1}(j;p)};\boldsymbol{\psi}) - \sum_{j=2}^D\log f(\boldsymbol{z}_{S_{d-1}(j;p)};\boldsymbol{\psi}),
\end{align}
where $D$ composite likelihood weights $w_S$ in \eqref{eq:composite} are set to $1$, $D-1$ weights are set to $-1$, and the rest are set to zero. There are thus only $2D-1$ likelihood terms to evaluate in \eqref{eq:compositeVecchia}, as opposed to $\sum_{d=1}^D{D\choose d}=2^D-1$ terms in \eqref{eq:composite} and ${D\choose d}$ terms in \eqref{eq:composited}. The dimension of densities involved in \eqref{eq:compositeVecchia} is at most $d$, and thus, is in some sense comparable to \eqref{eq:composited}. We write $\widehat{\boldsymbol{\psi}}_{V;d}$ to denote the mode of $\ell_{V;d}(\boldsymbol{\psi})=\sum_{i=1}^n \ell_{V;d}(\boldsymbol{\psi};\boldsymbol{z}_i)$, with $\ell_{V;d}(\boldsymbol{\psi};\boldsymbol{z})$ defined in \eqref{eq:compositeVecchia}, and because of the analogy between \eqref{eq:compositeVecchia} and \eqref{eq:composite}, the same asymptotic theory applies, although $\widehat{\boldsymbol{\psi}}_{V;d}$ usually provides gains in efficiency as compared to $\widehat{\boldsymbol{\psi}}_{C;d}$; see Sections~\ref{sec:Gaussian} and \ref{sec:MaxStable}.

Notice that because the Vecchia approximation relies on a \emph{nested} sequence of conditional events, the expression \eqref{eq:densityVecchia} is by construction a valid likelihood function that corresponds to a specific data generating process \citep{Katzfuss.Guinness:2021}, as opposed to pairwise likelihoods or more general composite likelihoods as in \eqref{eq:composite}. As such, it avoids using ``redundant'' information, which is key to improving the estimator's efficiency. As illustrated in Figure~\ref{fig:VecchiaApproximationProcess}, the Vecchia likelihood approximation actually yields an approximation of the process itself. 
\begin{figure}[t!]
\centering
\includegraphics[width=\linewidth]{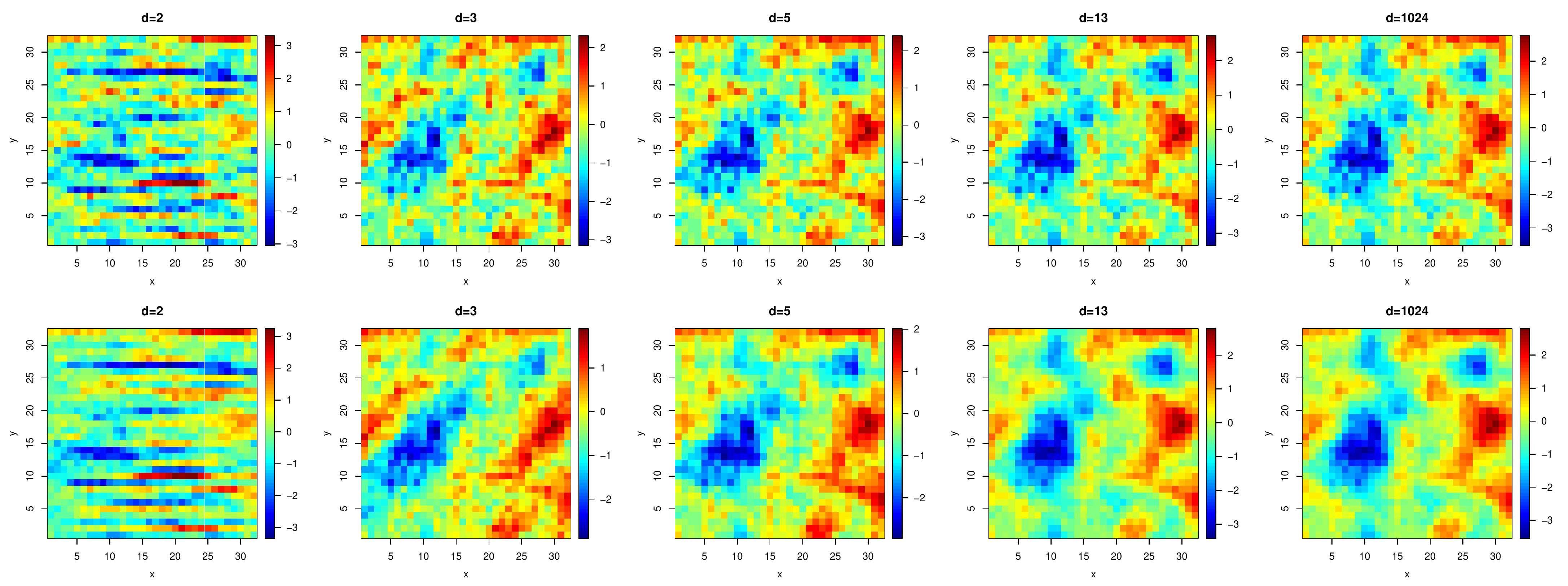}
\caption{Realizations from Gaussian processes on $\{1,\ldots,32\}^2$ with zero mean, unit variance, and correlation function $\corr\{Z(\boldsymbol{s}),Z(\boldsymbol{s}+\boldsymbol{h})\}=\exp(-\|\boldsymbol{h}\|/5)$ (top right) and $\corr\{Z(\boldsymbol{s}),Z(\boldsymbol{s}+\boldsymbol{h})\}=\exp\{-(\|\boldsymbol{h}\|/5)^{1.5}\}$ (bottom right), and their corresponding Vecchia approximations for $d=2,3,5,13$ (from left to right), using a coordinate-based ordering.}\label{fig:VecchiaApproximationProcess}
\end{figure}
The larger the cutoff dimension $d$, the better the approximation, as expected. For small cutoff dimensions $d$, the approximation fails at accurately representing the full joint distribution, although it captures the low-dimensional interactions reasonably well. The choice of a coordinate-based ordering for the Vecchia approximation is apparent for $d=2$, but the approximation improves dramatically as $d$ increases. In fact, since the Vecchia likelihood approximation is a valid likelihood function (thus, a density), it is possible to measure the quality the approximation by considering the Kullback-Leibler (KL) divergence of $f_{V;d}(\boldsymbol{z})$ with respect to the true likelihood $f(\boldsymbol{z})$ (with dependence on $\boldsymbol{\psi}$ suppressed for readability), i.e., ${\rm KL}(f\| f_{V;d})=\int f(\boldsymbol{z})\log\{f(\boldsymbol{z})/f_{V;d}(\boldsymbol{z})\}{\rm d}\boldsymbol{z}$; see, e.g., \citet{Schafer.etal:2021} for some approximation results in the Gaussian case. When subsets $S_{d-1}(j;p)$ are chosen as the nearest neighbors from the $j$-th site, we can show that, in the general case, ${\rm KL}(f\| f_{V;d})$ is always a non-increasing function of $d$, i.e., the approximate Vecchia likelihood gets ``closer and closer'' to the true likelihood, as expected. This result is formalized in Proposition~\ref{prop:KL}. Notice that this usually not does hold for general (renormalized) composite likelihoods. 
\begin{prop}\label{prop:KL}
Consider the true likelihood $f(\boldsymbol{z})$ in \eqref{eq:fulldensity}, $\boldsymbol{z}\in\mathcal Z\subset\Real^D$, and the Vecchia likelihood approximation $f_{V;d}(\boldsymbol{z})$ in \eqref{eq:densityVecchia}--\eqref{eq:compositeVecchia}, constructed from subsets $S_{d-1}(j;p)\subset H(j;p)$ based on some permutation $p$ and comprising the $\min(j,d)-1$ nearest neighbors of the $j$-th location (from its history $H(j;p)$). Then, the function $d\mapsto {\rm KL}(f\| f_{V;d})$ is monotone non-increasing in the cutoff dimension $d$. Moreover, when $d=D$, one has ${\rm KL}(f\| f_{V;d})=0$.
\end{prop}
\begin{proof}
By definition, one has 
\begin{align*}
{\rm KL}(f\| f_{V;d})&=\int f(\boldsymbol{z})\log\{f(\boldsymbol{z})/f_{V;d}(\boldsymbol{z})\}{\rm d}\boldsymbol{z}=h[f_{V;d}]-h[f],
\end{align*}
where $h[f]=-\int f(\boldsymbol{z})\log f(\boldsymbol{z}){\rm d}\boldsymbol{z}$ is the entropy of the density $f$, and similarly for $h[f_{V;d}]$. Since $h[f]$ is constant in the cutoff dimension $d$, it is sufficient to show that $h[f_{V;d}]\leq h[f_{V;d-1}]$ for all $d=3,\ldots,D$. By definition of the Vecchia approximation in \eqref{eq:densityVecchia}--\eqref{eq:compositeVecchia}, we can write
\begin{align*}
h[f_{V;d}]&=-\int f(\boldsymbol{z})\log f_{V;d}(\boldsymbol{z}){\rm d}\boldsymbol{z}\\
&=-\int f(\boldsymbol{z})\log\left\{f(z_{p(1)})\right\}{\rm d}\boldsymbol{z}-\sum_{j=2}^D\int f(\boldsymbol{z})\log\left\{f(z_{p(j)}\mid\boldsymbol{z}_{S_{d-1}(j;p)})\right\}{\rm d}\boldsymbol{z}\\
&=h[f_{Z_{p(1)}}]+\sum_{j=2}^Dh[f_{Z_{p(j)}\mid\boldsymbol{Z}_{S_{d-1}(j;p)}}],
\end{align*}
where $f_{Z_{p(1)}}$ denotes the density of the random variable $Z_{p(1)}\sim f(z_{p(1)})$, and $f_{Z_{p(j)}\mid\boldsymbol{Z}_{S_{d-1}(j;p)}}$ denotes the conditional density $f(z_{p(j)}\mid\boldsymbol{z}_{S_{d-1}(j;p)})$ of the random variable $Z_{p(j)}$ given $\boldsymbol{Z}_{S_{d-1}(j;p)}=\boldsymbol{z}_{S_{d-1}(j;p)}$. Now, because the subsets $S_{d-1}(j;p)$ are composed of nearest neighbors of the $j$-th variable, they are nested, i.e.,
$$S_{1}(j;p)\subset S_{2}(j;p)\subset\cdots\subset S_{D-1}(j;p)=H(j;p).$$
This implies that for each cutoff dimension $d=3,\ldots,D$, the conditioning variables $\boldsymbol{Z}_{S_{d-1}(j;p)}$ are the same as $\boldsymbol{Z}_{S_{d-2}(j;p)}$ but augmented with one additional variable. Since the conditional entropy $h[f_{X\mid Y}]$ is always smaller than or equal to the marginal entropy $h[f_{X}]$ for all random vectors $(X,Y)^\top\sim f_{X,Y}(x,y)$ (with equality if $X$ and $Y$ are independent), it follows that $h[f_{Z_{p(j)}\mid\boldsymbol{Z}_{S_{d-1}(j;p)}}]\leq h[f_{Z_{p(j)}\mid\boldsymbol{Z}_{S_{d-2}(j;p)}}]$, and thus $h[f_{V;d}]\leq h[f_{V;d-1}]$, for all $d=3,\ldots,D$. This proves that ${\rm KL}(f\| f_{V;d})$ is monotone non-increasing in $d$. Moreover, since $f_{V;D}=f$, we have that ${\rm KL}(f\| f_{V;D})=0$ by definition, which concludes the proof.
\end{proof}
The illustration in Figure~\ref{fig:VecchiaApproximationProcess} and the result in Proposition~\ref{prop:KL} both imply that the approximation improves as $d$ increases. This suggests that a similar improvement is to be expected in terms of the relative efficiency of the corresponding Vecchia likelihood estimator, $\widehat{\boldsymbol{\psi}}_{V;d}$.

Although the Vecchia log-likelihood in \eqref{eq:compositeVecchia} is appealing and has good efficiency, there is no reason why the corresponding weights $w_S\in\{-1,0,1\}$ should necessarily be optimal. Therefore, we also explore here a modified Vecchia likelihood obtained by changing the weights attributed to the conditioning sets, i.e., 
\begin{equation}\label{eq:compositemodifiedVecchia}
\ell_{V;d;\omega}(\boldsymbol{\psi};\boldsymbol{z})=\log f(\boldsymbol{z}_{p(1)};\boldsymbol{\psi}) + \sum_{j=2}^D \log f(z_{p(j)},\boldsymbol{z}_{S_{d-1}(j;p)};\boldsymbol{\psi}) + \omega\sum_{j=2}^D\log f(\boldsymbol{z}_{S_{d-1}(j;p)};\boldsymbol{\psi}),
\end{equation}
where $\omega\in[-1,\infty)$ is a weight to be selected. When $\omega=-1$, \eqref{eq:compositemodifiedVecchia} reduces to the Vecchia likelihood in \eqref{eq:compositeVecchia}, and when $\omega=0$, \eqref{eq:compositemodifiedVecchia} almost corresponds to a composite likelihood of order $d$ in \eqref{eq:composited}, with weights appropriately chosen. As $\omega\to\infty$, the contribution of the conditioning set dominates, and \eqref{eq:compositemodifiedVecchia} therefore roughly corresponds to a composite likelihood estimator of order $d-1$ with weights appropriately chosen. We write $\widehat{\boldsymbol{\psi}}_{V;d;\omega}$ to denote the mode of $\ell_{V;d;\omega}(\boldsymbol{\psi})=\sum_{i=1}^n \ell_{V;d;\omega}(\boldsymbol{\psi};\boldsymbol{z}_i)$, with $\ell_{V;d;\omega}(\boldsymbol{\psi};\boldsymbol{z})$ defined in \eqref{eq:compositemodifiedVecchia}. Higher efficiency can be obtained by fine-tuning the weight $\omega$. In the Supplementary Material, we do an in-depth investigation of the optimal choice of $\omega$ in the Gaussian setting, and we find that in general the classical Vecchia estimator with $\omega=-1$ is quite competitive in terms of its efficiency compared to the optimal case. In the sequel, we shall therefore set $\omega=-1$.

\section{Asymptotic relative efficiency in the Gaussian case}\label{sec:Gaussian}
\subsection{Setting}\label{subsec:setting}
In order to have a better theoretical understanding of the relative efficiencies of the different estimators introduced in Section~\ref{sec:composite}, we start by considering the Gaussian setting, which also provides qualitative insights into the behavior of these estimators in more complex settings. The max-stable case is studied in more detail by simulation in Section~\ref{sec:MaxStable}. 
Here, we consider a stationary Gaussian process $Z(\boldsymbol{s})$, $\boldsymbol{s}\in\Real^2$, with zero mean and unit variance, and we assume that data $\boldsymbol{Z}=(Z(\boldsymbol{s}_1),\ldots,Z(\boldsymbol{s}_D))^\top$ are located on the grid $\{1,\ldots,\sqrt{D}\}^2$ with $D=100$. To be concise, we here only consider an exponential spatial correlation model, while in the Supplementary Material we also investigate asymptotic relative efficiencies in a non-spatial, fully exchangeable model, as well as a powered exponential spatial correlation model.

We compare the (theoretical) asymptotic relative efficiency of the following estimators:
\begin{enumerate}
\item The maximum full likelihood estimator, denoted $\widehat{\boldsymbol{\psi}}$.
\item The composite likelihood estimator of order $d$, $\widehat{\boldsymbol{\psi}}_{C;d}$, defined in \eqref{eq:composited}. We consider the dimensions $d=2,3,4,5$ and adopt a truncation strategy as in \citet{Castruccio.etal:2016} to reduce the computational burden by setting the weights as $w_S=\mathbb I(\max_{\{i,j\}\subset S}\|\boldsymbol{h}_{\{i,j\}}\|\leq \delta)$ with cutoff distance $\delta=1,\sqrt{2},2,\sqrt{5},\sqrt{8}$ (i.e., selecting only the $1$st--$5$th-order neighbors, respectively). The number of selected likelihood terms in each case is reported in Table~\ref{tab:likelihoodterms}. For fixed $d$, this is roughly proportional to the time to compute $\widehat{\boldsymbol{\psi}}_{C;d}$.
\begin{table}[t!]
\centering
\caption{Number of likelihood terms involved in \eqref{eq:composited}, with $d=2,3,4,5$ and weights $w_S=\mathbb I(\max_{\{i,j\}\subset S}\|\boldsymbol{h}_{\{i,j\}}\|\leq \delta)$ with cutoff distance $\delta=1,\sqrt{2},2,\sqrt{5},\sqrt{8}$. The numbers below are for data sampled on the grid $\{1,\ldots,\sqrt{D}\}^2$ with $D=100$. Numbers in brackets are the proportions among the ${D\choose d}$ possible terms. The estimator $\widehat{\boldsymbol{\psi}}_{C;d}$ cannot be computed when the number of terms is zero. For comparison, the number of likelihood terms involved in the Vecchia likelihood \eqref{eq:compositeVecchia} is always $2D-1=199$.}\label{tab:likelihoodterms}
\vspace{5pt}
\begin{tabular}{c|rl|rl|rl|rl|rl|}
$d\setminus\delta$ & \multicolumn{2}{c|}{$1$} & \multicolumn{2}{c|}{$\sqrt{2}$} & \multicolumn{2}{c|}{$2$} & \multicolumn{2}{c|}{$\sqrt{5}$} & \multicolumn{2}{c|}{$\sqrt{8}$} \\ \hline
$2$ & $180$ & $(3.64\%)$ & $342$ & $(6.91\%)$ & $502$ & $(10.14\%)$ & $790$ & $(15.96\%)$ & $918$ & $(18.55\%)$  \\
$3$ & $0$ & $(0\%)$ & $324$ & $(0.20\%)$ & $772$ & $(0.48\%)$ & $2436$ & $(1.51\%)$ & $3332$ & $(2.06\%)$ \\
$4$ & $0$ & $(0\%)$ & $81$ & $(10^{-3}\%)$ & $433$ & $(0.01\%)$ & $3809$ & $(0.10\%)$ & $6433$ & $(0.16\%)$ \\
$5$ & $0$ & $(0\%)$ & $0$ & $(0\%)$ & $64$ & $(10^{-4}\%)$ & $3232$ & $(10^{-3}\%)$ & $7392$ & $(0.01\%)$ 
\end{tabular}
\end{table}
\item The Vecchia likelihood estimator, $\widehat{\boldsymbol{\psi}}_{V;d}$, defined in \eqref{eq:compositeVecchia}. 
We consider $d=2,3,4,5,9,13,21$ 
and select the $d-1$ nearest neighbors in the ``past'' variables. We compare the four different orderings of variables considered by \citet{Guinness:2018}: the coordinate-based ordering ($p_1$), a random ordering ($p_2$), the middle-out ordering ($p_3$), and the maximum-minimum ordering ($p_4$). The middle-out ordering starts with the variable at the center of the grid (which minimizes the average distance to all other points), and then selects the order of variables according to their distance to the center point. The maximum-minimum ordering also starts from the center variable, but then selects the next variables in a way that maximizes the minimum distance to all previously selected points. If there are multiple points that maximize the minimum distance, 
we select the next variable randomly among the possible solutions. The different orderings are illustrated in Figure~\ref{fig:orderings}.
\end{enumerate}
\begin{figure}[t!]
\centering
\includegraphics[width=\linewidth]{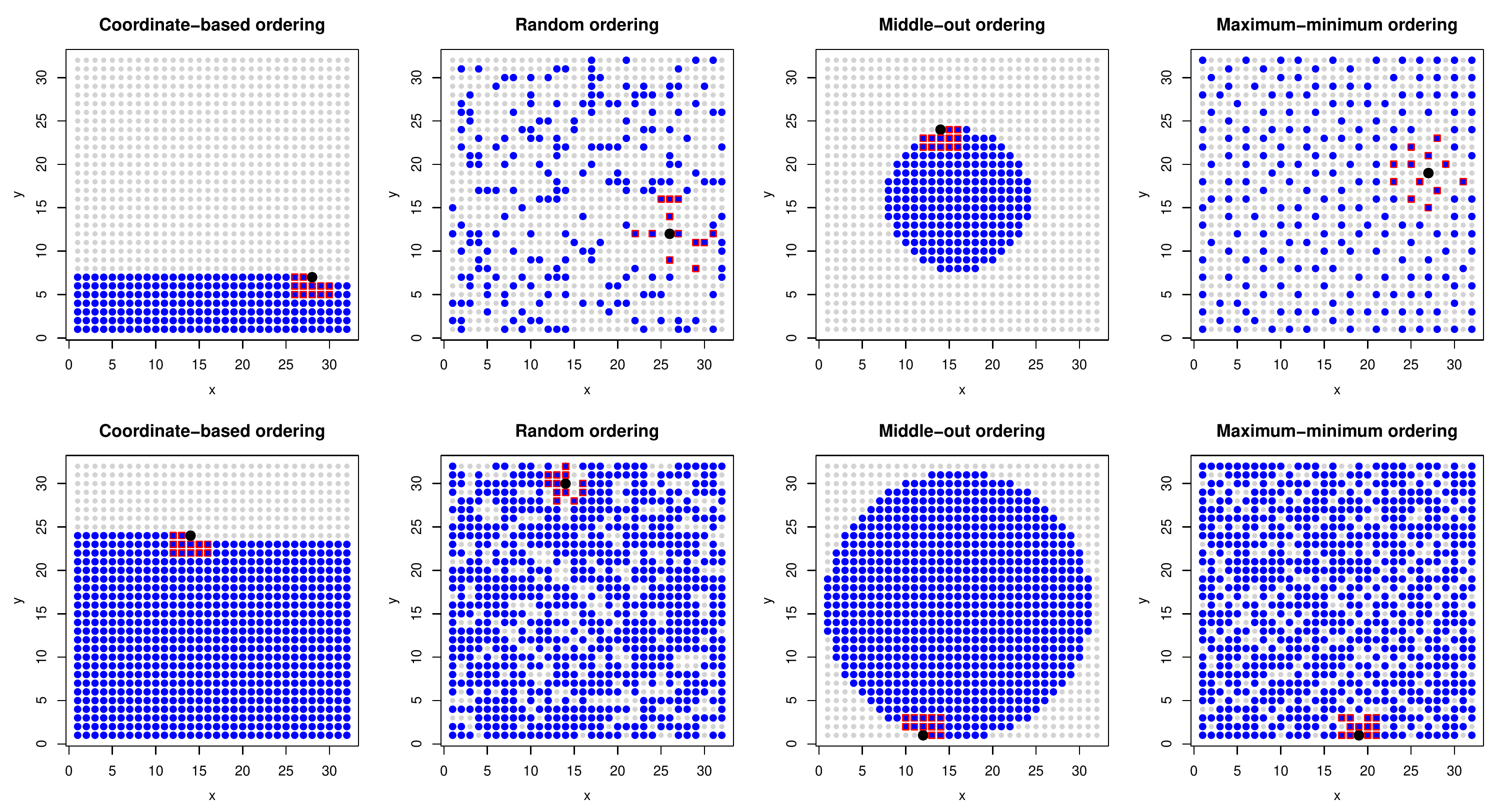}
\caption{Illustration of the four different orderings considered for the Vecchia approximation (reproduced from \citealp{Guinness:2018}): coordinate-based (left), random (2nd column), middle-out (3rd column) and maximum-minimum (right), when data are assumed to be sampled on the grid $\{1,\ldots,32\}^2$ (grey dots). The black dots represent the 220th (top) and 750th (bottom) points for each ordering. The blue dots represent the ``past'' variables, and the red squares are the $12$ nearest neighbors among the ``past'' variables.}\label{fig:orderings}
\end{figure}

The asymptotic relative efficiency of an estimator $\widehat{\boldsymbol{\psi}}_A$ (either $\widehat{\boldsymbol{\psi}}_{C;d}$, $\widehat{\boldsymbol{\psi}}_{V;d}$, or $\widehat{\boldsymbol{\psi}}_{V;d;\omega}$) with respect to the maximum full likelihood estimator $\widehat{\boldsymbol{\psi}}$ is defined as follows. Let $\boldsymbol{V}_A$ and $\boldsymbol{V}$ be the corresponding asymptotic variance matrices. The exact formula for the asymptotic variance matrices are provided in Appendix~\ref{subsec:general}. For the $r$th parameter, we then define the marginal relative efficiency as the ratio of asymptotic standard deviations, i.e., ${\rm ARE}(\widehat{\boldsymbol{\psi}}_{A;r})=(\boldsymbol{V}_{r,r}/\boldsymbol{V}_{A;r,r})^{1/2}$. The overall relative efficiency is defined as ${\rm ARE}(\widehat{\boldsymbol{\psi}}_A)=(|\boldsymbol{V}|/|\boldsymbol{V}_{A}|)^{1/(2q)}$. When $\boldsymbol{\psi}$ is a scalar (i.e., $m=1$), the two definitions coincide.

\subsection{Results based on the exponential correlation function}\label{subsec:exponential}
We here study a spatial model with exponential correlation function $\corr\{Z(\boldsymbol{s}),Z(\boldsymbol{s}+\boldsymbol{h})\}=\exp(-\|\boldsymbol{h}\|/\lambda)$, where $\boldsymbol{h}$ is the spatial lag vector, $\|\boldsymbol{h}\|$ is its length, and $\boldsymbol{\psi}\equiv\lambda>0$ is the range parameter. The larger $\lambda$, the stronger the spatial correlation. 

\begin{figure}[t!]
\centering
\includegraphics[width=0.49\linewidth]{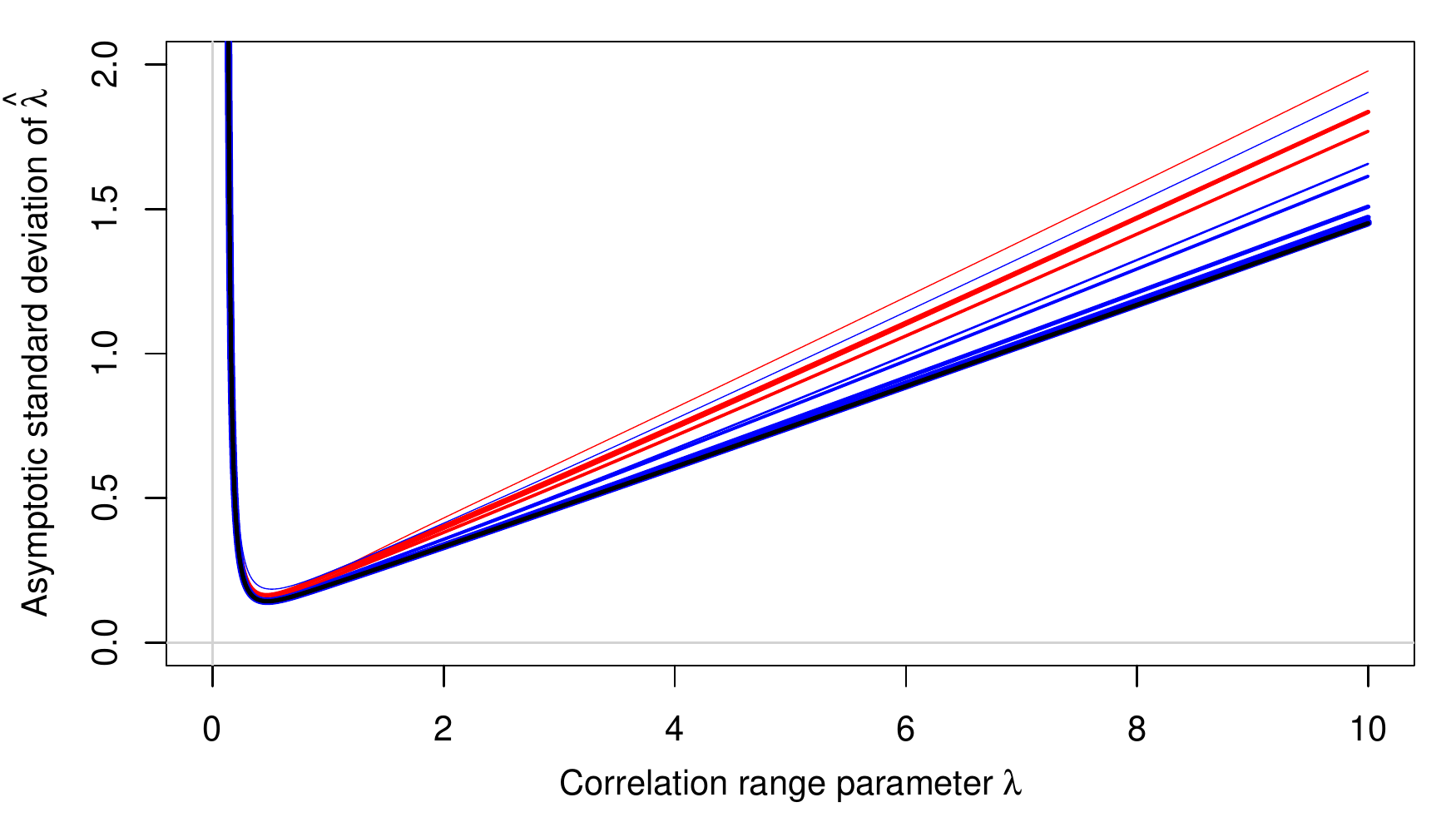}
\includegraphics[width=0.49\linewidth]{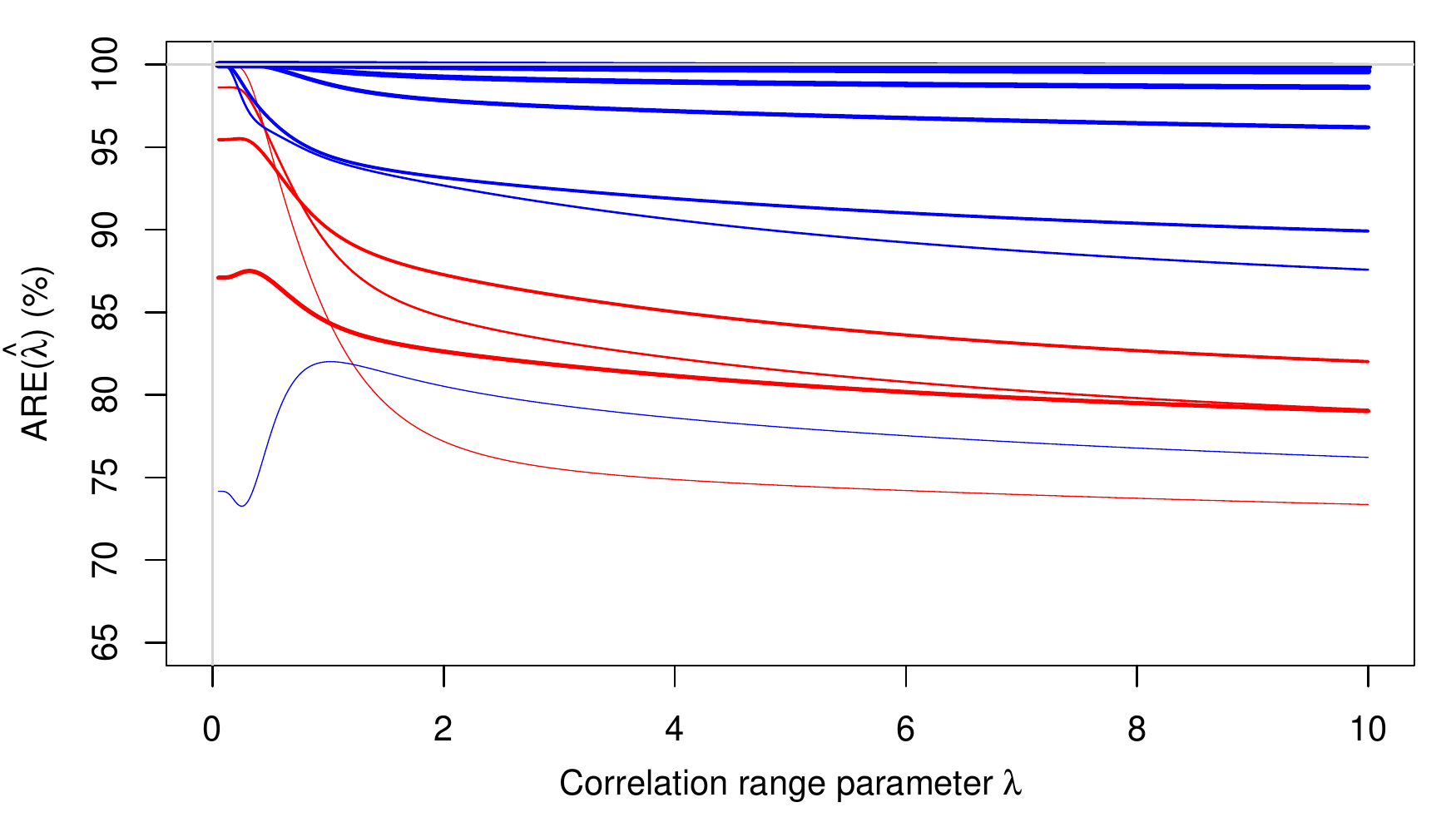}
\caption{Asymptotic standard deviation (left) and asymptotic relative efficiency (right) of the full likelihood estimator $\widehat{\lambda}$ (in black), the composite likelihood estimator $\widehat{\lambda}_{C;d}$ (in red) with $d=2,3,4,5$ (thin to thick curves) and cutoff distance $\delta=2$ (keeping about $10\%$ of pairs), and the Vecchia likelihood estimator $\widehat{\lambda}_{V;d}$ (in blue) with $d=2,3,4,5,9,13,21$ (thin to thick curves) and based on a coordinate ordering. We consider here the exponential model $\corr\{Z(\boldsymbol{s}),Z(\boldsymbol{s}+\boldsymbol{h})\}=\exp(-\|\boldsymbol{h}\|/\lambda)$, with true value $\lambda\in(0,10)$.}\label{fig:EfficiencyExponential}
\end{figure}

Figure~\ref{fig:EfficiencyExponential} displays the asymptotic standard deviation and asymptotic relative efficiency of the composite likelihood estimator $\widehat{\lambda}_{C;d}$ with cutoff distance $\delta=2$ (keeping about $10\%$ of pairs) and the Vecchia likelihood estimator $\widehat{\lambda}_{V;d}$ using a coordinate-based ordering, as a function of the range parameter $\lambda$, for various choices of $d$ (the dimension of likelihood terms). The Vecchia estimator $\widehat{\lambda}_{V;d}$ largely outperforms $\widehat{\lambda}_{C;d}$ for most values of $d$ and $\lambda$. Almost perfect efficiency is attained by the Vecchia likelihood estimator $\widehat{\lambda}_{V;d}$ for $d\geq5$.

\begin{table}[t!]
\centering
\caption{Asymptotic relative efficiency ($\%$) of the composite likelihood estimator $\widehat{\lambda}_{C;d}$ (left) with $d=2,3,4,5$ and cutoff distance $\delta=1,\sqrt{2},2,\sqrt{5},\sqrt{8}$, and of the Vecchia likelihood estimator $\widehat{\rho}_{V;d}$ (right) with $d=2,3,4,5,9,13,21$ and coordinate-based ($p_1$), random ($p_2$, middle out ($p_3$) and maximum-minimum ($p_4$) orderings. We consider here the exponential model $\corr\{Z(\boldsymbol{s}),Z(\boldsymbol{s}+\boldsymbol{h})\}=\exp(-\|\boldsymbol{h}\|/\lambda)$, with true value $\lambda=5$.}\label{tab:CompositeVecchiaExponential}
\vspace{5pt}
\begin{tabular}{r|>{\centering}p{1.3cm}>{\centering}p{1.3cm}>{\centering}p{1.3cm}>{\centering}p{1.3cm}>{\centering}p{1.3cm}|>{\centering}p{1.3cm}>{\centering}p{1.3cm}>{\centering}p{1.3cm}>{\centering\arraybackslash}p{1.3cm}|}
& \multicolumn{5}{c|}{Composite estimator $\widehat{\rho}_{C;d}$} & \multicolumn{4}{c|}{Vecchia estimator $\widehat{\rho}_{V;d}$}\\
& \multicolumn{5}{c|}{Cutoff distance $\delta$} & \multicolumn{4}{c|}{Ordering}\\
$d$ & $1$ & $\sqrt{2}$ & $2$ & $\sqrt{5}$ & $\sqrt{8}$ & $p_1$ & $p_2$ & $p_3$ & $p_4$ \\ \hline
$2$ & $90.4$ & $82.6$ & $74.5$ & $64.8$ & $60.6$ & $78.0$ & $67.7$ & $75.3$ & $58.7$ \\
$3$ & --- & $86.8$ & $81.4$ & $72.3$ & $69.3$ & $89.9$ & $83.4$ & $90.2$ & $84.5$ \\
$4$ & --- & $90.5$ & $84.3$ & $78.3$ & $75.8$ & $91.4$ & $93.1$ & $92.6$ & $92.4$ \\
$5$ & --- & --- & $80.6$ & $82.4$ & $80.4$ & $97.0$ & $96.8$ & $97.1$ & $97.7$ \\
$9$ &  &  &  &  &  & $98.9$ & $99.0$ & $99.5$ & $99.4$ \\
$13$ &  &  &  &  &  & $99.7$ & $99.7$ & $99.9$ & $99.8$ \\
$21$ &  &  &  &  &  & $99.9$ & $100.0$ & $100.0$ & $100.0$ 
\end{tabular}
\end{table}

Table~\ref{tab:CompositeVecchiaExponential} reports the asymptotic relative efficiency of the composite likelihood estimator $\widehat{\lambda}_{C;d}$ and the Vecchia likelihood estimator $\widehat{\lambda}_{V;d}$ for $\lambda=5$ and various choices of cutoff dimension $d$, cutoff distance $\delta$ and ordering. When $d=2$, $\widehat{\lambda}_{C;d}$ generally performs better than $\widehat{\lambda}_{V;d}$, but when $d>2$, the Vecchia likelihood estimator $\widehat{\lambda}_{V;d}$ has in most cases a better efficiency than $\widehat{\lambda}_{C;d}$. The gains are even (much) more substantial for the exchangeable model studied in the Supplementary Material. Counter-intuitively, the performance of the composite likelihood estimator generally has a worse performance for larger cutoff distances $\delta$, which is due to the re-use of information when including many similar (and highly dependent) likelihood terms in \eqref{eq:composited}. For example, when $d=2$, the relative efficiency of $\widehat{\lambda}_{C;d}$ is about $90\%$ for $\delta=1$ but only $60\%$ when $\delta=\sqrt{8}$. Moreover, it is not always true that the $\widehat{\lambda}_{C;d}$ has a better performance as $d$ increases (for fixed $\delta$); see the results for the powered exponential model in the Supplementary Material for an example. 
By contrast, the Vecchia likelihood estimator $\widehat{\lambda}_{V;d}$ is always found to have a better performance as $d$ increases (for fixed ordering), as expected from Proposition~\ref{prop:KL}.

\section{Simulation study in the max-stable case}\label{sec:MaxStable}
\subsection{Max-stable models}\label{subsec:maxstable}
As already noted, max-stable processes are the only possible limits of suitably renormalized pointwise maxima of independent and identically distributed random fields. 
More specifically, let $Y_1(\boldsymbol{s}),Y_2(\boldsymbol{s}),\ldots$ denote independent copies of the random field $Y(\boldsymbol{s})$, $\boldsymbol{s}\in\Real^2$, and let $M_n(\boldsymbol{s})=\max\{Y_1(\boldsymbol{s}),\ldots,Y_n(\boldsymbol{s})\}$ be the process of pointwise maxima. Furthermore, assume that $Y(\boldsymbol{s})$ satisfies the max-domain of attraction condition, i.e., there exist sequences $a_n(\boldsymbol{s})>0$ and $b_n(\boldsymbol{s})$ such that
\begin{equation}\label{eq:maxstable}
a_n^{-1}(\boldsymbol{s})\{M_n(\boldsymbol{s})-b_n(\boldsymbol{s})\}\Dto Z(\boldsymbol{s}),
\end{equation}
where the convergence holds in the sense of finite-dimensional distributions and the limit process $Z(\boldsymbol{s})$ has non-degenerate margins. Then, $Z(\boldsymbol{s})$ is a max-stable process, with generalized extreme-value (GEV) marginal distributions, and $Y(\boldsymbol{s})$ is said to be in the max-domain of attraction of $Z(\boldsymbol{s})$. Upon marginal transformation, we can assume without loss of generality that $Z(\boldsymbol{s})$ has unit Fr\'echet margins, i.e., $\pr\{Z(\boldsymbol{s})\leq z\}=\exp(-1/z)$, $z>0$. On the unit Fr\'echet scale, the max-stability property implies that for each $t>0$, and every finite collection of sites $\{\boldsymbol{s}_1,\ldots,\boldsymbol{s}_D\}\subset\Real^2$,
\begin{equation}\label{eq:maxstable2}
\pr\{Z(\boldsymbol{s}_1)\leq tz_1,\ldots,Z(\boldsymbol{s}_D)\leq tz_D\}^t=\pr\{Z(\boldsymbol{s}_1)\leq z_1,\ldots,Z(\boldsymbol{s}_D)\leq z_D\}.
\end{equation}
Thanks to \citet{deHaan:1984}'s representation, max-stable processes may be constructed as follows. Let $W_1(\boldsymbol{s}),W_2(\boldsymbol{s}),\ldots$ be independent copies of a non-negative process $W(\boldsymbol{s})$ with unit mean, and let $\xi_1,\xi_2,\ldots$ be points of a Poisson process with intensity $\xi^{-2}{\rm d}\xi$ on $(0,+\infty)$. Then the process defined as
\begin{equation}\label{eq:maxstable3}
Z(\boldsymbol{s})=\sup_{i=1,2,\ldots} \xi_i W_i(\boldsymbol{s})
\end{equation}
is a max-stable process with unit Fr\'echet margins and finite-dimensional distributions
\begin{equation}\label{eq:maxstable4}
\pr\{Z(\boldsymbol{s}_1)\leq z_1,\ldots,Z(\boldsymbol{s}_D)\leq z_D\}=\exp\{-V(z_1,\ldots,z_D)\}:=\exp\{-V(\boldsymbol{z})\},
\end{equation}
where the exponent function $V$ may be written in terms of the $W$ process as $V(\boldsymbol{z}):=V(z_1,\ldots,z_D)=\E[\max\{W(\boldsymbol{s}_1)/z_1,\ldots,W(\boldsymbol{s}_D)/z_D\}]$, $\boldsymbol{z}=(z_1,\ldots,z_D)^\top$. 
In particular, $V$ is homogeneous of order $-1$, i.e., $V(tz_1,\ldots,tz_D)=t^{-1}V(z_1,\ldots,z_D)$ for all $t>0$, and satisfies $V(z,\infty,\ldots,\infty)=1/z$ for any permutation of the arguments. 

To construct useful max-stable models, the challenge is to find flexible processes $W(\boldsymbol{s})$, for which the exponent function $V$ can be computed. Our simulation results below are based on the Brown--Resnick model \citep{Kabluchko.etal:2009}, which is a popular model for spatial extremes. In the Supplementary Material, we also explore the case of the multivariate logistic max-stable model \citep{Gumbel:1960,Gumbel:1961}, which is exchangeable in all variables.

From \eqref{eq:maxstable4}, the joint density of a parametric max-stable process may be expressed as 
\begin{equation}\label{eq:densitymaxstable}
f(\boldsymbol{z};\boldsymbol{\psi})=\exp\{-V(\boldsymbol{z};\boldsymbol{\psi})\}\sum_{\pi\in \mathcal P_D}\prod_{\tau\in \pi}\left\{-V_\tau(\boldsymbol{z};\boldsymbol{\psi})\right\},
\end{equation} 
where $\mathcal P_D$ is the collection of all partitions $\pi=\{\tau_1,\ldots,\tau_{|\pi|}\}$ of the set $\{1,\ldots,D\}$ (of cardinality $|\pi|$), $V_\tau$ denotes the partial derivative of the function $V$ with respect to the variables indexed by the set $\tau\subseteq\{1,\ldots,D\}$, and $\boldsymbol{\psi}\in\Psi\subseteq\Real^q$ denotes the vector of parameters; see \citet{Huser.etal:2016}, \citet{Castruccio.etal:2016} and \citet{Huser.etal:2019}. Because the number of terms in the sum on the right-hand side of \eqref{eq:densitymaxstable} is the Bell number, which grows more than exponentially with $D$, it is not possible to perform full likelihood inference for max-stable processes observed in moderate or high dimensions. \citet{Huser.etal:2019} proposed a stochastic EM-estimator but its applicability is still limited to relatively small dimensions (i.e., $D\leq 20$) for the Brown--Resnick model and similar max-stable models; see also \citet{Thibaud.etal:2016} and \citet{Dombry.etal:2017} for a similar inference approach from a Bayesian perspective. \citet{Padoan.etal:2010} proposed using a pairwise likelihood with weights appropriately chosen, while \citet{Castruccio.etal:2016} investigated the gains in efficiency of higher-order truncated composite likelihoods of the form \eqref{eq:composited} with $d\geq2$. In our simulations below, as well as in the Supplementary Material, we demonstrate that considerable efficiency gains can be obtained with the Vecchia approximation \eqref{eq:compositeVecchia} in most cases, while being scalable to high dimensions. 

\subsection{Results for the Brown--Resnick model}\label{subsec:BrownResnick}
We now consider the popular Brown--Resnick spatial process \citep{Kabluchko.etal:2009} constructed as in \eqref{eq:maxstable3}, where $W$ is a log-Gaussian process defined as
\begin{equation}\label{eq:BR}
W(\boldsymbol{s})=\exp\{\varepsilon(\boldsymbol{s})-\sigma(\boldsymbol{s})^2/2\},
\end{equation}
with $\sigma(\bm s)>0$  and $\varepsilon(\bm s)$ a Gaussian process with mean zero and variance $\sigma(\bm s)^2$. By analogy with the Gaussian exponential correlation model in Section~\ref{subsec:exponential}, we here explore the case where $\varepsilon(\boldsymbol{s})$ is stationary with exponential correlation function $\rho(\boldsymbol{h})=\exp(-\|\boldsymbol{h}\|/\lambda)$, $\lambda>0$, and $\sigma(\boldsymbol{s})\equiv \sigma>0$, although it would also possible to consider more complex Gaussian processes with stationary increments. When the Brown--Resnick process is observed at the sites $\boldsymbol{s}_1,\ldots,\boldsymbol{s}_D\in\mathcal S$, the corresponding exponent function may be written as
\begin{equation}\label{eq:BRV}
V(z_1,\ldots,z_D;\boldsymbol{\psi})=\sum_{j=1}^D{1\over z_j}\Phi_{D-1}\left(\boldsymbol{\eta}_j;\boldsymbol{\Sigma}_j\right),
\end{equation}
where the parameter vector is here $\boldsymbol{\psi}=(\lambda,\sigma)^\top\in\Psi=(0,+\infty)^2$, $\Phi_{D-1}(\cdot;\boldsymbol{\Sigma})$ denotes the $(D-1)$-dimensional Gaussian distribution with zero mean vector and covariance matrix $\boldsymbol{\Sigma}$, $\boldsymbol{\eta}_j$ is a $(D-1)$-dimensional vector with $i$th component $\log(z_i/z_j)/\Gamma_{ij}^{1/2} + \Gamma_{ij}^{1/2}/2$, $i\neq j$, and $\boldsymbol{\Sigma}_j$ is a $(D-1)\times(D-1)$ matrix with $(i_1,i_2)$-entry $(\Gamma_{i_1j}+\Gamma_{i_2j}-\Gamma_{i_1i_2})/\{2(\Gamma_{i_1j}\Gamma_{i_2j})^{1/2}\}$, $i_1,i_2\neq j$, where $\Gamma_{ij}=\Gamma(\boldsymbol{s}_j-\boldsymbol{s}_i)$ and $\Gamma(\boldsymbol{h})$ denotes the underlying variogram function, here equal to $\Gamma(\boldsymbol{h})=2\sigma^2\{1-\rho(\boldsymbol{h})\}$; see \citet{Huser.Davison:2013a} and \citet{Wadsworth.Tawn:2014}. Partial and full derivatives of \eqref{eq:BRV} needed for (composite) likelihood computations (recall \eqref{eq:densitymaxstable}) may be found in \citet{Wadsworth.Tawn:2014}. 

A dependence summary that is well suited for max-stable processes is the extremal coefficient. Considering two sites $\boldsymbol{s},\boldsymbol{s}+\boldsymbol{h}\in\mathcal S$ at spatial lag $\boldsymbol{h}$, the extremal coefficient is defined through the exponent function $V$ (restricted to these two sites) as
\begin{equation}\label{eq:extremalcoef}
\theta(\boldsymbol{h})=V(1,1;\boldsymbol{\psi})=2\Phi\{\Gamma^{1/2}(\boldsymbol{h})/2\}=2\Phi([2\sigma^2\{1-\rho(\boldsymbol{h})\}]^{1/2}/2),
\end{equation}
where $\Phi(\cdot)$ is the univariate Gaussian distribution function. When $\theta(\boldsymbol{h})=1$, the corresponding pair of max-stable variables $\{Z(\boldsymbol{s}),Z(\boldsymbol{s}+\boldsymbol{h})\}$ are perfectly dependent, and when $\theta(\boldsymbol{h})=2$ they are completely independent. Therefore, complete independence cannot be captured unless $\sigma\to\infty$. Alternative unbounded variograms, e.g., $\Gamma(\boldsymbol{h})=(\|\boldsymbol{h}\|/\lambda)^\alpha$ with $\lambda>0,\alpha\in(0,2]$, allow for complete independence as $\|\boldsymbol{h}\|\to\infty$.

In our simulations, we sample data at $D=100$ locations on the grid $\{1,\ldots, \sqrt{D}\}^2$, with $n=100$ independent replicates.
We fix $\sigma=10$ and consider $\lambda=1,\ldots,10$ (short to long range dependence), which yields the extremal coefficient functions plotted in Figure~\ref{fig:extremalcoefficient}. 
\begin{figure}[t!]
\centering
\includegraphics[width=0.49\linewidth]{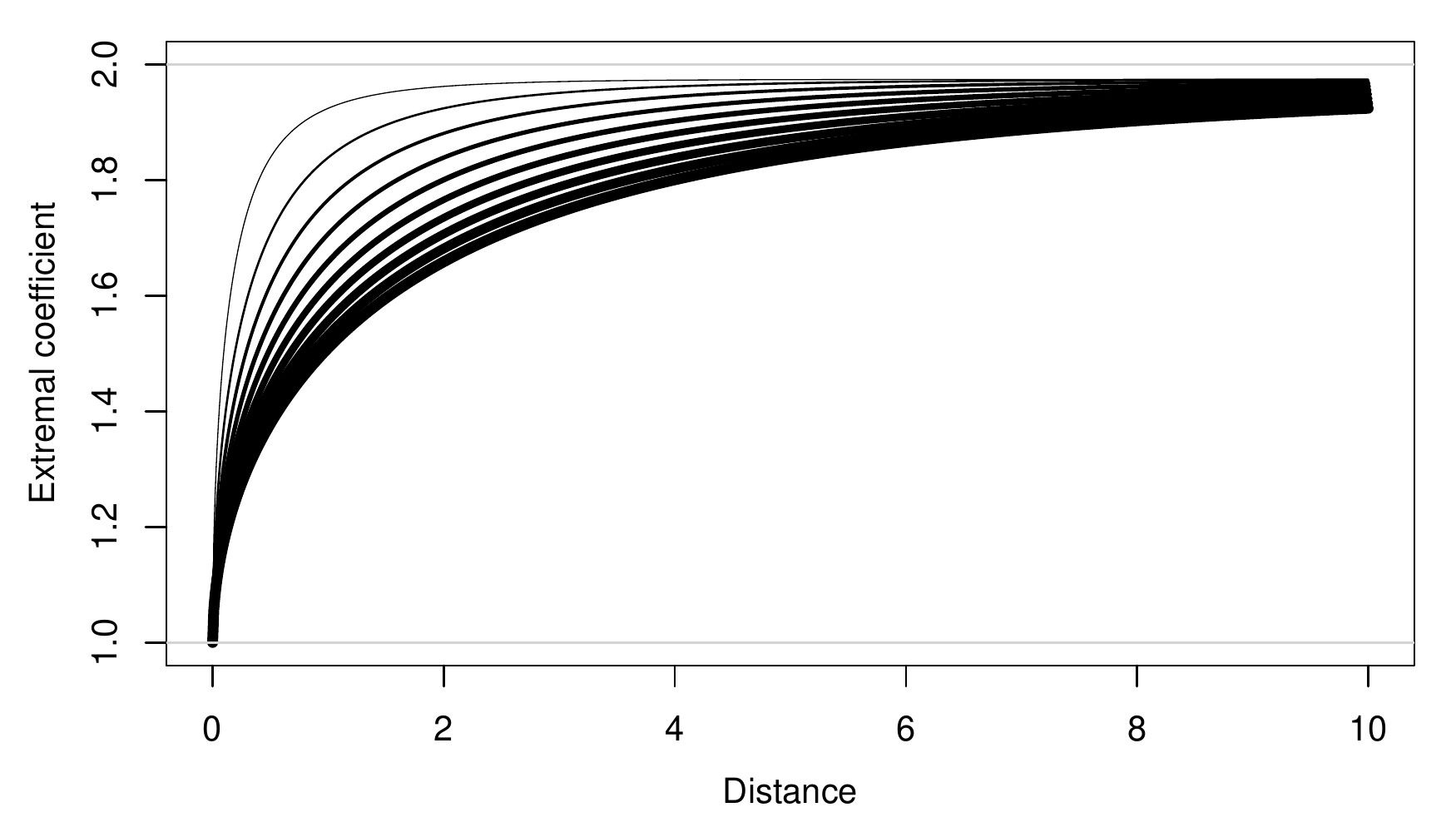}
\caption{Extremal coefficient curves for the Brown--Resnick model with variogram $\Gamma(\boldsymbol{h})=2\sigma^2\{1-\rho(\boldsymbol{h})\}$ and $\sigma=10$, $\lambda=1,\ldots,10$ (thin to thick curves).}\label{fig:extremalcoefficient}
\end{figure}
For each simulated dataset, we then estimate the range parameter $\lambda$ (treating $\sigma$ as known) using the composite likelihood estimators $\widehat\lambda_{C;d}$ and Vecchia likelihood estimators $\widehat\lambda_{V;d}$ described in Section~\ref{subsec:setting}, except that here we restrict ourselves to cutoff distances $\delta=1,\sqrt{2},2$, and cutoff dimensions $d=2,3,4,5$ for computational reasons. Larger values of $\delta$ and $d$ are considered for the logistic max-stable model in the Supplementary Material. We repeated the experiments $1024$ times to compute the estimators' bias, standard deviation and root mean squared error (RMSE).

\begin{figure}[t!]
\centering
\includegraphics[width=0.49\linewidth]{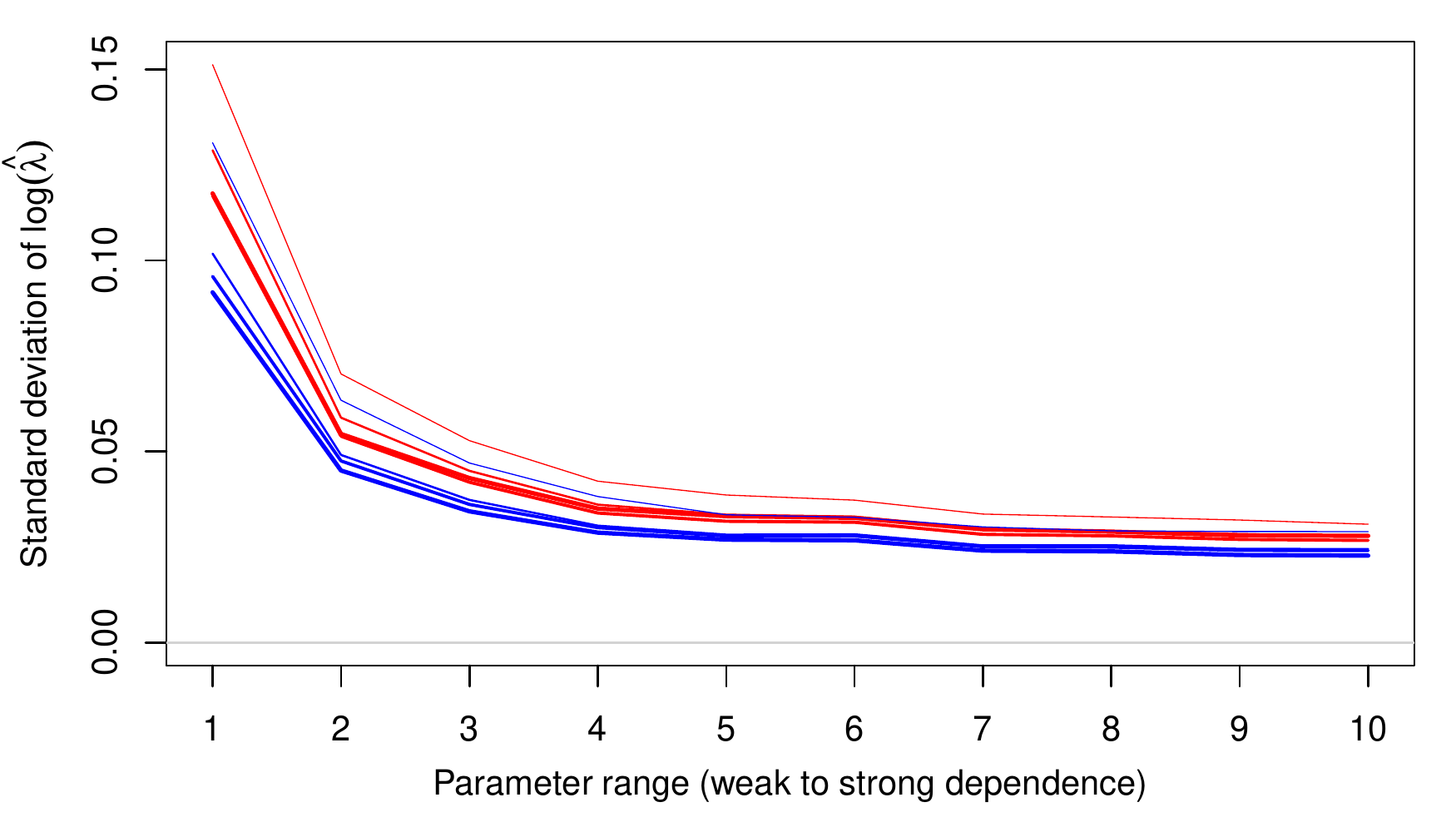}
\includegraphics[width=0.49\linewidth]{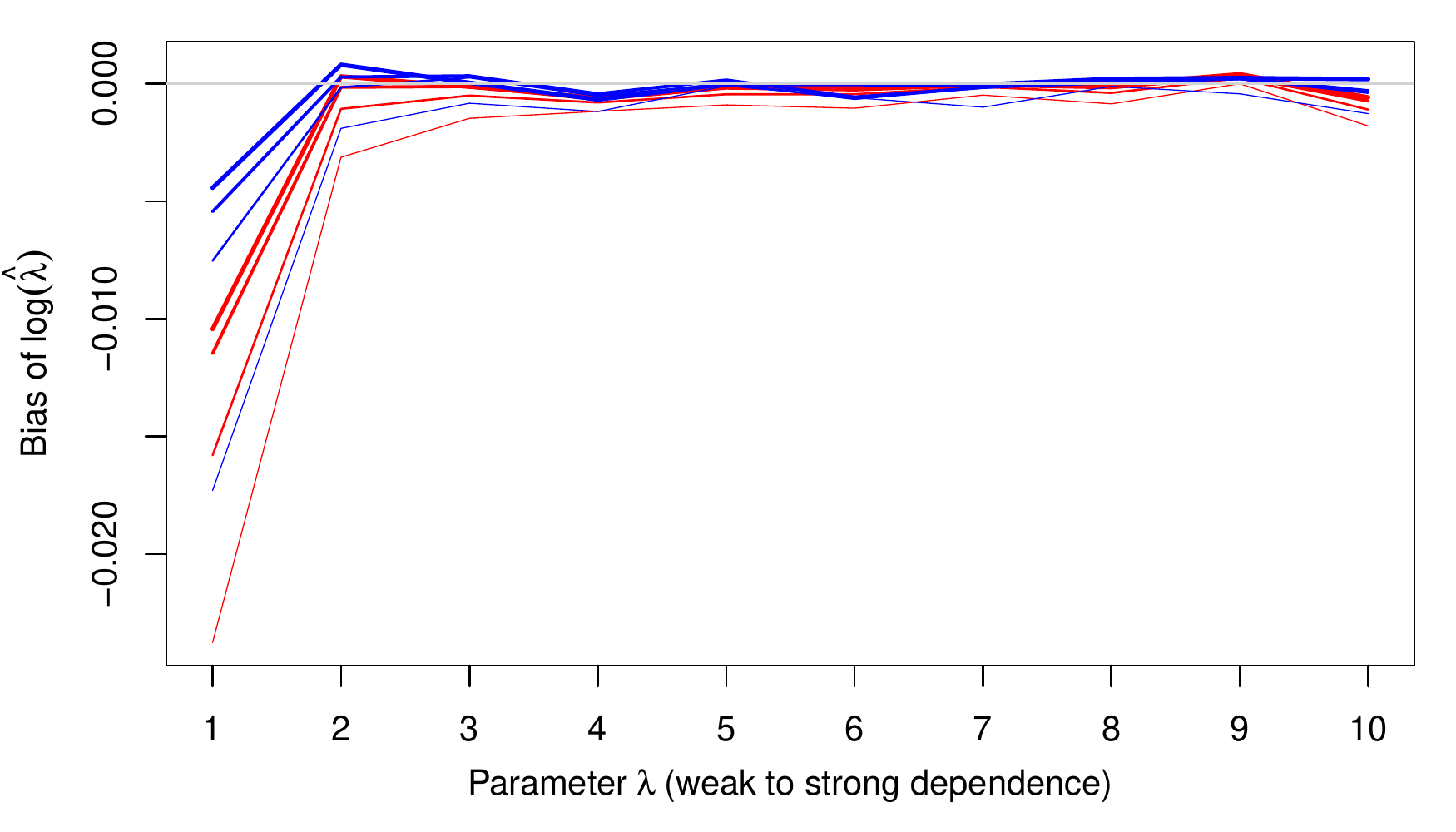}
\includegraphics[width=0.49\linewidth]{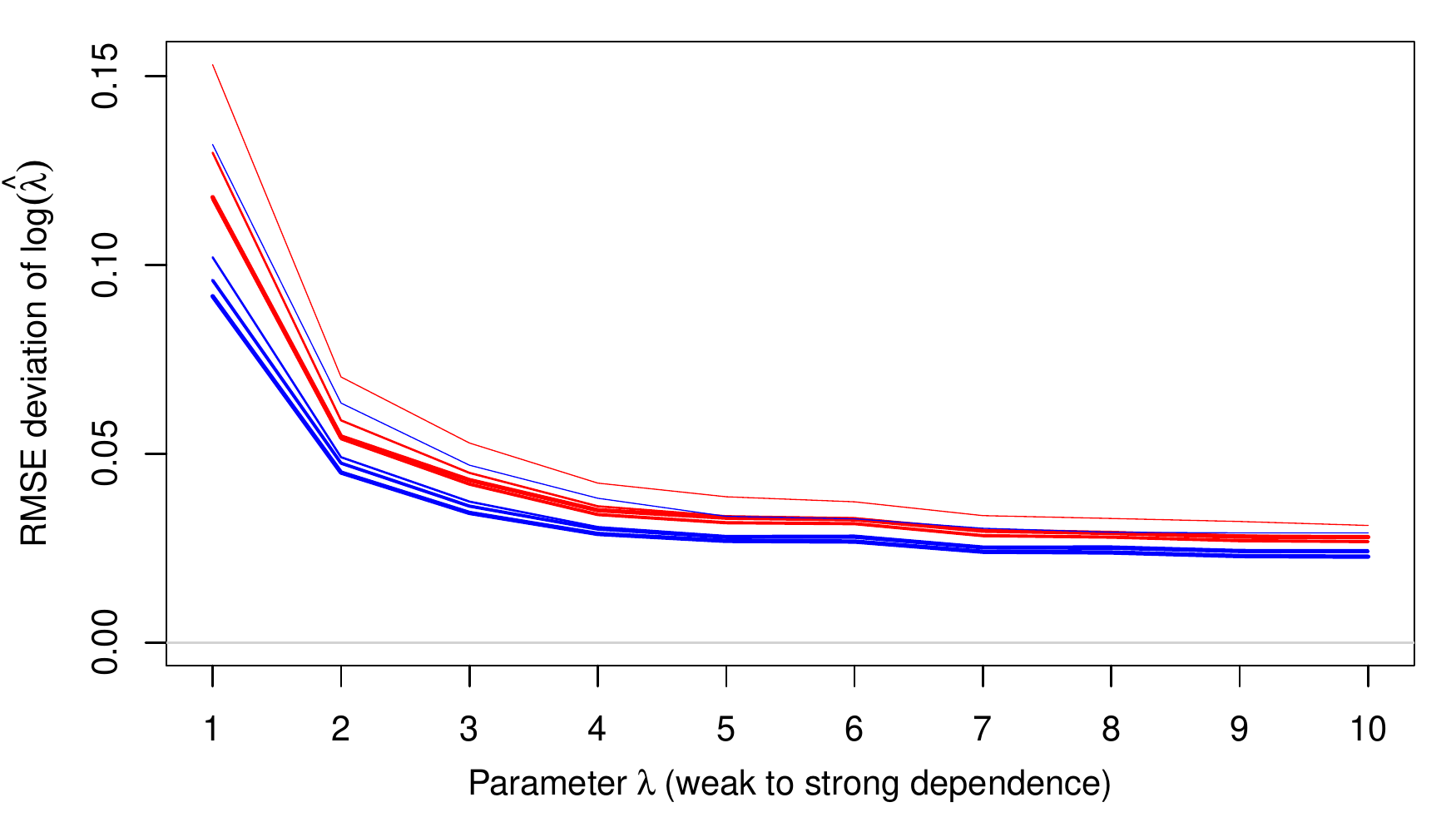}
\includegraphics[width=0.49\linewidth]{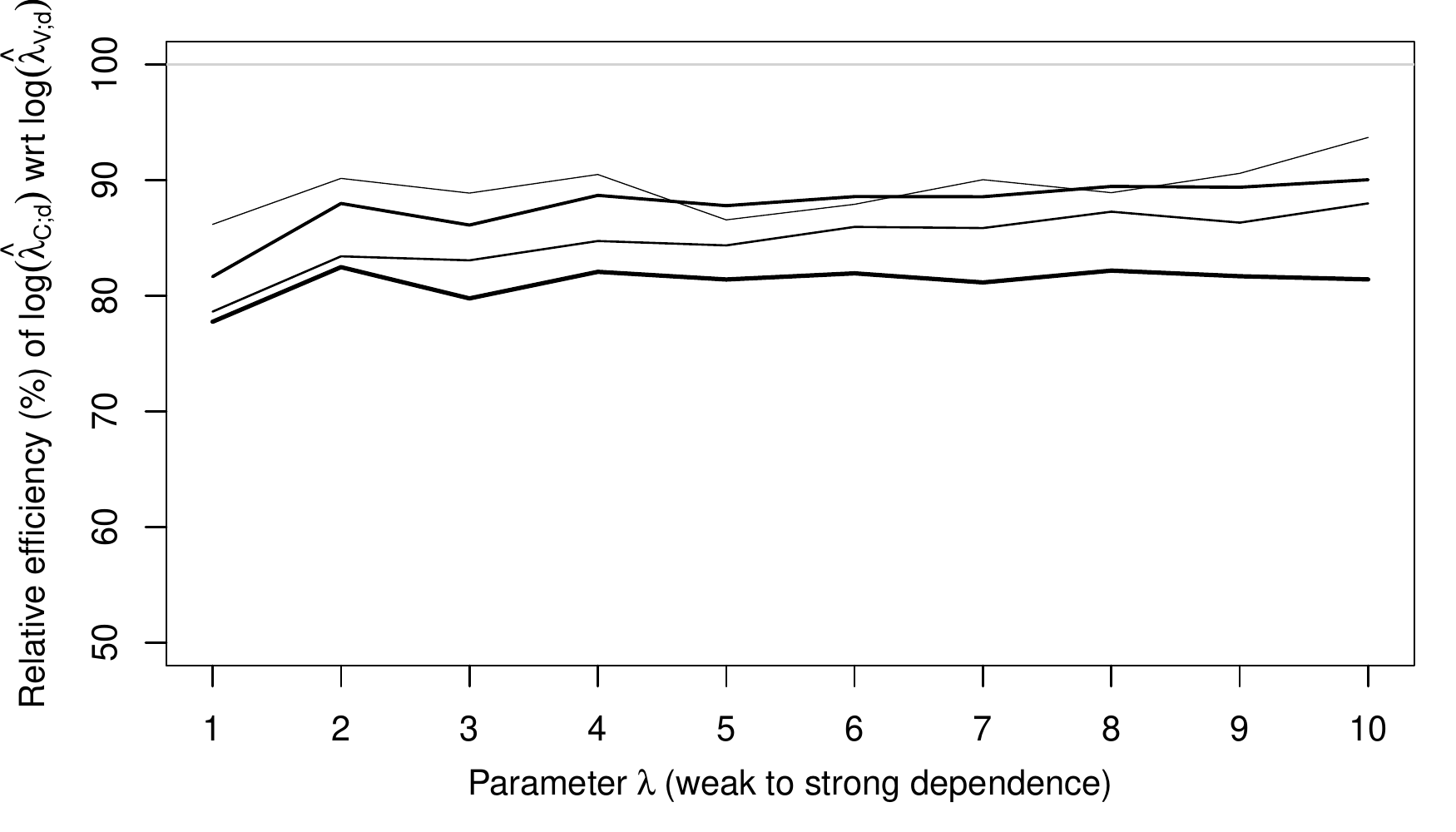}
\caption{Standard deviation (top left), bias (top right) and root mean squared error (bottom left) of the composite likelihood estimator $\log(\widehat{\lambda}_{C;d})$ (in red) with $d=2,3,4,5$ (thin to thick curves) and cutoff distance $\delta=2$ (keeping about $10\%$ of pairs), and the Vecchia likelihood estimator $\log(\widehat{\lambda}_{V;d})$ (in blue) with $d=2,3,4,5$ (thin to thick curves) and based on a coordinate ordering. The bottom right panel shows the relative efficiency of $\log(\widehat{\lambda}_{C;d})$ with respect to $\log(\widehat{\lambda}_{V;d})$ for $d=2,3,4,5$ (thin to thick curves). We consider here the Brown--Resnick model with parameters $\sigma=10$ and $\lambda=1,\ldots,10$ (weak to strong dependence).}\label{fig:EfficiencyBRMaxStable}
\end{figure}
The results are summarized in Figure~\ref{fig:EfficiencyBRMaxStable} (with $\delta=2$ for $\widehat\lambda_{C;d}$ and coordinate-based ordering for $\widehat\lambda_{V;d}$). Essentially, the bias of all estimators is negligible compared to the standard deviation, and the Vecchia likelihood estimator $\log(\widehat\lambda_{V;d})$ is about $10$--$20\%$ more efficient than the composite likelihood estimator $\log(\widehat\lambda_{C;d})$ for any dimension $d$ and range parameter $\lambda$ (with the efficiency defined as the ratio of RMSEs). The RMSE of all estimators with other cutoff distances $\delta$ and orderings is reported in Table~\ref{tab:CompositeVecchiaBRMaxStable}. The results are consistent with our previous theoretical findings in the Gaussian case, i.e., the Vecchia likelihood estimator always has higher efficiency than the composite likelihood estimator, except in the case with $d=2$ and $\delta=1$. Moreover, the Vecchia likelihood estimator performs better with the coordinate or middle-out orderings.

\begin{table}[t!]
\centering
\caption{Root mean squared error ($\times100$) for the composite likelihood estimator $\log(\widehat{\lambda}_{C;d})$ (left) with $d=2,3,4,5$ and cutoff distance $\delta=1,\sqrt{2},2$, and of the Vecchia likelihood estimator $\log(\widehat{\lambda}_{V;d})$ (right) with $d=2,3,4,5$ and coordinate-based ($p_1$), random ($p_2$), middle out ($p_3$) and maximum-minimum ($p_4$) orderings. We consider here the max-stable Brown--Resnick model with parameters $\sigma=10$ and $\lambda=5$, simulated in dimension $D=100$ with $n=100$ replicates.}\label{tab:CompositeVecchiaBRMaxStable}
\vspace{5pt}
\begin{tabular}{r|>{\centering}p{1.3cm}>{\centering}p{1.3cm}>{\centering}p{1.3cm}|>{\centering}p{1.3cm}>{\centering}p{1.3cm}>{\centering}p{1.3cm}>{\centering\arraybackslash}p{1.3cm}|}
& \multicolumn{3}{c|}{Composite estimator $\log(\widehat{\lambda}_{C;d})$} & \multicolumn{4}{c|}{Vecchia estimator $\log(\widehat{\lambda}_{V;d})$}\\
& \multicolumn{3}{c|}{Cutoff distance $\delta$} & \multicolumn{4}{c|}{Ordering}\\
$d$ & $1$ & $\sqrt{2}$ & $2$ & $p_1$ & $p_2$ & $p_3$ & $p_4$ \\ \hline
$2$ & $3.08$ & $3.43$ & $3.86$ & $3.34$ & $3.93$ & $3.56$ & $4.54$ \\
$3$ & --- & $3.08$ & $3.34$ & $2.82$ & $3.05$ & $2.83$ & $3.18$ \\
$4$ & --- & $2.93$ & $3.17$ & $2.79$ & $2.81$ & $2.77$ & $2.87$ \\
$5$ & --- & --- & $3.31$ & $2.69$ & $2.71$ & $2.68$ & $2.70$ \\
\end{tabular}
\end{table}

The computational time of each estimator is reported in Table~\ref{tab:CompositeVecchiaBRMaxStableTime}. We also provide an estimate of the computational time for the composite likelihood estimator $\widehat\lambda_{C;d}$ with $\delta=\sqrt{5},\sqrt{8}$ by extrapolating the times obtained with $\delta=2$ by assuming that these are proportional to the number of likelihood terms reported in Table~\ref{tab:likelihoodterms}. While the computational time for $\widehat\lambda_{C;d}$ grows fast as a function of the cutoff distance $\delta$, it is essentially the same for each ordering considered for $\widehat\lambda_{V;d}$. Moreover, the computational time remains fairly moderate as $d$ increases for $\widehat\lambda_{V;d}$, but it can be extremely large for $\widehat\lambda_{C;d}$ when $d=4,5$. 

Overall, when $d=2$, the best solution is to use $\widehat\lambda_{C;d}$ but when $d>2$ the best solution is to use $\widehat\lambda_{V;d}$ for reasons of both statistical efficiency and computational efficiency.

\begin{table}[t!]
\centering
\caption{Computational time (hr) for the composite likelihood estimator $\widehat{\lambda}_{C;d}$ (left) with $d=2,3,4,5$ and cutoff distance $\delta=1,\sqrt{2},2,\sqrt{5},\sqrt{8}$, and of the Vecchia likelihood estimator $\log(\widehat{\lambda}_{V;d})$ (right) with $d=2,3,4,5$ and coordinate-based ($p_1$), random ($p_2$), middle out ($p_3$) and maximum-minimum ($p_4$) orderings. We consider here the max-stable Brown--Resnick model with parameters $\sigma=10$ and $\lambda=5$, simulated in dimension $D=100$ with $n=100$ replicates. Numbers with an asterisk are extrapolated from the $\delta=2$ case by assuming that the computational time for $\widehat{\lambda}_{C;d}$ is proportional to the numbers reported in Table~\ref{tab:likelihoodterms}.}\label{tab:CompositeVecchiaBRMaxStableTime}
\vspace{5pt}
\begin{tabular}{r|>{\centering}p{1.3cm}>{\centering}p{1.3cm}>{\centering}p{1.3cm}>{\centering}p{1.3cm}>{\centering}p{1.4cm}|>{\centering}p{1.3cm}>{\centering}p{1.3cm}>{\centering}p{1.3cm}>{\centering\arraybackslash}p{1.3cm}|}
& \multicolumn{5}{c|}{Composite estimator $\widehat{\lambda}_{C;d}$} & \multicolumn{4}{c|}{Vecchia estimator $\widehat{\lambda}_{V;d}$}\\
& \multicolumn{5}{c|}{Cutoff distance $\delta$} & \multicolumn{4}{c|}{Ordering}\\
$d$ & $1$ & $\sqrt{2}$ & $2$ & $\sqrt{5}$ & $\sqrt{8}$ & $p_1$ & $p_2$ & $p_3$ & $p_4$ \\ \hline
$2$ & $0.047$ & $0.097$ & $0.153$ & $0.241^\star$ & $0.280^\star$ & $0.026$ & $0.025$ & $0.025$ & $0.026$ \\
$3$ & --- & $0.706$ & $1.704$ & $5.376^\star$ & $7.353^\star$ & $0.229$ & $0.228$ & $0.235$ & $0.223$ \\
$4$ & --- & $0.633$ & $3.349$ & $29.463^\star$ & $49.761^\star$ & $1.100$ & $1.023$ & $1.111$ & $0.960$ \\
$5$ & --- & --- & $1.315$ & $66.394^\star$ & $151.852^\star$ & $3.386$ & $3.402$ & $3.398$ & $3.331$ \\
\end{tabular}
\end{table}

To investigate the scalability of the Vecchia likelihood estimator, we repeated our experiments for the Brown--Resnick model with parameters $\sigma=10$ and $\lambda=5$ in dimensions $D=25,49,100,144,225,400,625,1024$. 
Timing results reported in the Supplementary Material demonstrate that, as expected, the computational time is linear in $D$, but it grows fast in $d$. In fact, the time is roughly proportional to the Bell number of order $d$ (i.e., the cardinality of $\mathcal P_d$, the set of partitions of $\{1,\ldots,d\}$, recall \eqref{eq:densitymaxstable}). Nevertheless, with moderate values of $d$, the linearity in $D$ makes it possible to tackle high-dimensional extreme-value problems using the Vecchia likelihood estimator, while retaining fairly high efficiency.

Further simulations (not shown) show that similar results hold in the max-domain of attraction of the Brown--Resnick model, when simulating block maxima from the exponential factor copula model \citep{Krupskii.etal:2018,CastroCamilo.Huser:2019} with standard Pareto margins (i.e., when both the dependence structure and the marginal distributions are misspecified), with block size equal to $10^4$. When the block size is smaller, such as $100$ or $1000$, the sub-asymptotic bias is quite large but it is comparable across all estimators. Moreover, further results in the Supplementary Material show that for the exchangeable logistic max-stable model, even more substantial gains in efficiency can be obtained by considering the Vecchia likelihood estimator than reported here for the Brown--Resnick model.

\section{Data application}\label{sec:application}
\subsection{Red Sea surface temperature dataset}
The spatial modeling of sea surface temperature (SST) extremes plays a key role in estimating changes in the Earth's climate \citep{Bulgin.etal:2020} and understanding how ecosystems and marine life may be affected by global warming \citep{Tittensor.etal:2021}. While estimating marginal trends in SST observations is important for future predictions and risk planning and mitigation, characterizing their spatial tail dependence structure is needed to estimate extreme SST hotspots \citep{Hazra.Huser:2021}, and to assess the spatial extent of regions simultaneously affected by single extreme temperature events (see, e.g., \citealp{Zhong.etal:2021}). In our real data application, we analyze (standardized) annual maxima of SST anomalies for the whole Red Sea, obtained on a fine grid of 1043 locations for 31 years from 1987 to 2015. The spatial grid is displayed in Figure~\ref{fig:variogram_contour}. 
\begin{figure}[!t]
	\centering
       	\includegraphics[width=0.8\textwidth]{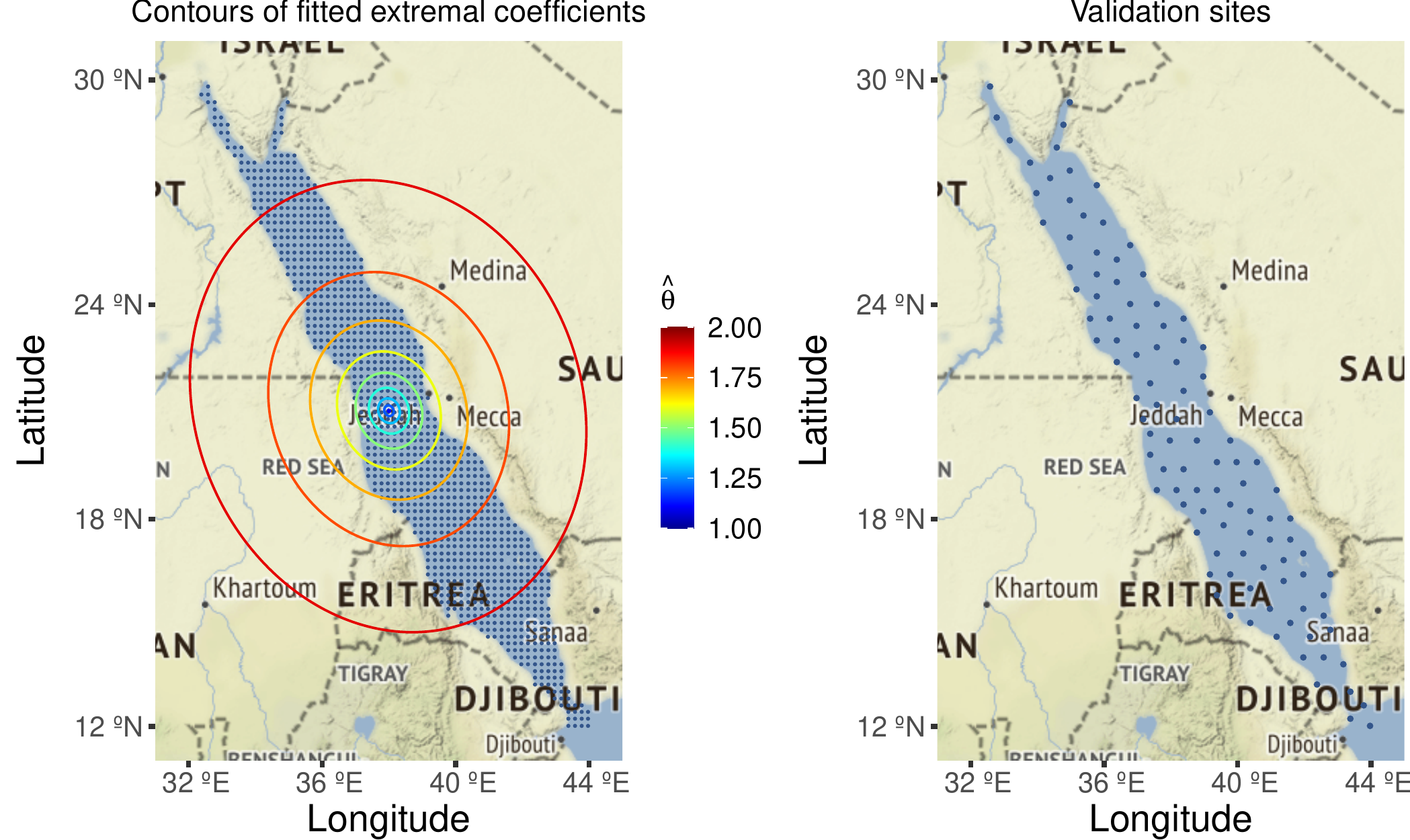}
		\caption{\emph{Left:} Map of the study domain, with the spatial grid (dots) covering the Red Sea at which SST data are available. Ellipses show contours of the fitted extremal coefficient function $\widehat\theta(\boldsymbol{h})=1.1,\ldots,1.9$, with respect to the grid cell at the center, obtained by fitting the anisotropic Brown--Resnick max-stable model using the best Vecchia likelihood estimator. \emph{Right:} Validation locations used in our cross-validation study.}
		\label{fig:variogram_contour}
\end{figure}
The Red Sea is a semi-enclosed sea with a very rich biodiversity, including abundant coral species that are often highly sensitive to modest SST increases. Before detailing our modeling of spatial extremal dependence, we first briefly summarize how the original data were pre-processed to obtain temperature anomalies, and how annual maxima thereof were then modeled and transformed to a common scale.  

The original data product was obtained from the Operational Sea Surface Temperature and Sea Ice Analysis (OSTIA; \citealp{donlon2012operational}), which produces satellite-derived daily SST measurements at a very high $0.05^\circ\times 0.05^\circ$ spatial resolution; see \citet{Huser:2021} for a detailed exploratory analysis of this dataset, and \citet{Hazra.Huser:2021} for a comprehensive spatial analysis. In our study, we subsampled the spatial locations while still maintaining good spatial coverage (i.e., keeping one measurement about every 18 kilometers in each direction), thus yielding $11{,}315$ fields of $D=1043$ highly spatially dependent daily observations, when discarding February $29$th in leap years to keep the same number of observations each year. Because daily temperature data feature seasonality, and a possible time trend due to global warming, which also varies across space, it is therefore crucial to first detrend the marginal distributions and standardize them to a common scale, before modeling dependencies among SST extremes with a max-stable process. To estimate spatiotemporal trends (both in the mean and the variance of daily temperatures) in a very flexible way, we fitted a semiparametric normal model to all temperature observations within a certain radius of each spatial location, using a local likelihood approach. This yields very accurate spatiotemporal trend estimates, due to our large sample size. Then, after standardizing the data based on the fitted semiparametric model, we extracted annual maxima of SST anomalies and fitted a generalized extreme-value (GEV) distribution, which we then used to transform annual SST maxima to a common unit Fr\'echet scale by means of the probability integral transform. For further details on marginal modeling, see the Supplementary Material.

In the next section, we model the dependence structure of the standardized annual maxima by fitting isotropic and anisotropic Brown--Resnick max-stable processes, and we focus on investigating differences between the performance of the traditional composite likelihood and the Vecchia likelihood approximation methods.

\subsection{Dependence modeling of the Red Sea temperature extremes}
To fit the max-stable Brown--Resnick model, we first need to specify the variogram function $\Gamma$ of the underlying Gaussian process $\varepsilon(\bm s)$ in \eqref{eq:BR}, which determines the form and range of dependencies that can be captured. In our simulation study, we used a bounded variogram of the form $\Gamma(\bm h)=2\sigma^2\{1-\rho(\bm h)\}\leq 2\sigma^2$, based on the stationary and isotropic exponential correlation function $\rho(\bm h)$, for comparison purposes with the Gaussian setting. Such a comparison is important to make sure the exact theoretical efficiency results in the Gaussian case (Section~\ref{sec:Gaussian}) can be generalized and carried over by analogy to the max-stable case (Section~\ref{sec:MaxStable}). However, using a bounded variogram also implies long-range dependence as the extremal coefficient is bounded away from independence at any spatial distance, i.e., $\theta(\bm h)<2$. This is problematic in our data application, since we model SST anomaly maxima over a very large domain, namely the whole Red Sea, for which complete independence prevails at large distances. This suggests that we should use an unbounded variogram in our application. Moreover, given that the Red Sea has a geographically elongated shape, that it is only connected to the World Ocean through the artificial Suez canal in the North and the Gulf of Aden in the South, and because of the complex hydrodynamic patterns that these physical constraints entail, it makes sense to use an anisotropic variogram function. Therefore, the variogram model that we use here is 
\begin{equation}
\Gamma(\bm s_1,\bm s_2) = \E\left[\left\{\varepsilon(\bm s_1) - \varepsilon(\bm s_2)\right\}^2\right] =2\left({\sqrt{(\bm s_1 - \bm s_2)^T A  (\bm s_1 - \bm s_2)}\over\lambda}\right)^\alpha, 	
\end{equation}
where $\lambda\in (0,\infty)$ is a range parameter, $\alpha \in (0,2)$ is a smoothness parameter, and $A$ is the rotation matrix, which has the form
\begin{equation}
	A =  \begin{bmatrix} \cos(\theta) & -\sin(\theta) \\  \sin(\theta) & \cos(\theta)  \end{bmatrix} \begin{bmatrix}
		1 & 0 \\ 0 & a \end{bmatrix} \begin{bmatrix} \cos(\theta) & \sin(\theta) \\  -\sin(\theta) & \cos(\theta)  \end{bmatrix},
\end{equation}
where $\theta\in \mathbb (-\pi/2,\pi/2)$ is the rotation angle, and $a>0$ determines the extent of anisotropy, with $a=1$ corresponding to isotropy. The dependence parameter vector, $\bm\psi$, thus consists of four parameters, i.e., $\bm \psi = (\alpha,\lambda,a,\theta)^\top\in\Psi=(0,2)\times(0,\infty)^2\times(-\pi/2,\pi/2)$. 

To fit the Brown--Resnick model, we consider the (traditional) weighted composite likelihood method, as well as the proposed Vecchia likelihood approximation, which is expected to boost both the computational and statistical efficiency according to the theoretical and simulation results reported in Sections~\ref{sec:Gaussian} and \ref{sec:MaxStable}. For the composite likelihood method, we consider pairwise ($d=2$) and triplewise ($d=3$) likelihoods, but cannot consider higher values of $d>3$ due to computational reasons. For each cutoff dimension $d$, we choose binary weights $\mathbb I(\max_{\{i,j\}}\|\bm h_{\{i,j\}}\|\leq \delta)$ as in Section~\ref{subsec:setting} with cutoff distance $\delta$ specified in such a way that the resulting composite likelihood function contains $m\times D$ terms in total, where $D=1043$ is the number of locations and $m=2,4,6,8$. Therefore, $m=2$ roughly corresponds to including $1$st-order neighbors only, $m=4$ roughly corresponds to including $2$nd-order neighbors only, and so forth, though the complex Red Sea boundaries mean that a few additional higher-order neighbors (i.e., at slightly longer distances) may also be included. For the Vecchia likelihood approximation, we consider the cutoff dimensions $d=2,3,4,5$ 
and use the orderings described in Section~\ref{subsec:setting}: coordinate-based ($p_1$), random ($p_2$), middle-out ($p_3$), and maximum-minimum ($p_4$).

Because the data are (approximately) gridded, there are only a few unique pairwise distances that characterize the likelihood contributions involved in the composite and Vecchia likelihoods. This, combined with the fact that SST maxima are highly spatially dependent, implies that the range parameter $\lambda$ and the smoothness parameter $\alpha$ may not be easily identifiable, and we have indeed found it difficult to estimate them both simultaneously. In our analysis, we thus fix the smoothness parameter to three representative values, i.e., $\alpha=0.5$ (rough field), $\alpha=1$ (intermediate case, similar to a Brownian motion), and $\alpha=1.5$ (smooth field), then estimate the parameter vector $\bm\psi_\alpha=(\lambda,a,\theta)^\top$ by maximizing the composite and Vecchia likelihoods for fixed $\alpha$, and subsequently select the best value of $\alpha$ by cross-validation. 

An extensive cross-validation study is hence conducted to compare the goodness-of-fit and prediction performance of the different fitted models, obtained by (i) varying the value of $\alpha\in\{0.5,1,1.5\}$; (ii) considering the general anisotropic Brown--Resnick model or its isotropic restriction (with $a=1$, $\theta=0$ fixed); and (iii) using different inference approaches (composite or Vecchia likelihoods under different settings). Precisely, we leave out a validation set consisting of about $10\%$ locations (i.e., exactly $105$ out of $1043$), selected as the last $10\%$ locations from the maximum-minimum ordering (recall Section~\ref{subsec:setting}). This ensures that the validation locations are well spread-out throughout the whole Red Sea; see Figure~\ref{fig:variogram_contour}. Then, we calculate the sum of the negative conditional log-density for each spatiotemporal point from the validation set, given the values at its four closest neighbors from the training set for the same temporal replicate. In other words, the (negative) log-score we consider is
\begin{equation}
\label{eq:score}
S = -\sum_{i=1}^n \sum_{j\in \calV} \log f(z_{i;j}\mid \bm z_{i;\calT(j)};\widehat{\bm\psi}_\alpha)
\end{equation}
where $\calV$ is the index set of validation locations, $\calT(j)$ is the index set of training locations that are the four closest neighbors of the $j$th location $\bm s_j$, $\bm z_{i;\calT(j)}$ is the corresponding observation vector from these neighboring locations from the training set, $f$ is the Brown--Resnick density, and $\widehat{\bm\psi}_\alpha$ is the estimated parameter vector (for fixed $\alpha$), obtained for each inference method. Since we consider only four nearest neighbors to calculate the score \eqref{eq:score}, the densities involved are of maximum dimension five, which is still computationally feasible.
 
The cross-validation results are reported in Table~\ref{tab:scores}. 
\begin{sidewaystable}
\caption{Negative conditional log score $S$ in \eqref{eq:score} for each scenario based on a selection of $10\%$ evenly spread validation locations. The three values in each cell correspond to $\alpha=0.5,1,1.5$, respectively. The best scenarios in each block (Vecchia/Composite $\times$ Anisotropic/Isotropic) are highlighted in bold. Orderings $p_1$, $p_2$, $p_3$, $p_4$ correspond to coordinate-based, random, middle-out, and maximum-minimum orderings, respectively. The maximum dimension of likelihood contributions is $d$ (i.e., for Vecchia methods, with at most $d-1$ variables in the conditioning sets), and $m$ relates to the cutoff distance for composite likelihoods.}
\label{tab:scores}
\begin{adjustbox}{width=\textheight}
\begin{tabular}{rccccc | ccl}
\toprule
& \multicolumn{5}{c|}{Vecchia} & \multicolumn{3}{c}{Composite} \\ \cmidrule(lr){2-6} \cmidrule(lr){7-9}
& Ordering & $d=2$ & $d=3$ & $d=4$ & $d=5$ & $d=2$ & $d=3$  & $m$ \\ \midrule
\multirow{4}{*}{Anisotropic}  & $p_1$ & $97.6/92.5/107.7$ & $96.6/91.6/107.1$ & $97.7/91.4/113$ & $98/91.4/112.7$ & $130.1/\bm{115.9}/1053.6$ & $125.8/117.6/946.3$ & $2$ \\ \cmidrule(lr){2-9}
& $p_2$ & $101.4/93.2/118$ & $97.7/91.6/105$ & $100.6/91.4/105.5$ & $97.3/91.5/105.1$ & $129.8/117/1101.1$ & $126/117/930.8$ & $4$ \\  \cmidrule(lr){2-9}
& $p_3$ & $99.2/93.2/107$ & $95.9/91.8/104.2$ & $97.2/91.4/106.1$ & $97.8/\bm{91.4}/105.2$ & $129.7/117.4/1097.3$ & $126.1/116.9/927.6$ & $6$ \\  \cmidrule(lr){2-9}
& $p_4$ & $102.8/93.6/126.1$ & $98/91.8/106.4$ & $101.4/91.5/104.9$ & $97.1/91.5/106.6$ & $129.7/117.1/1061.7$ & $126.3/116.9/929.5$ & $8$ \\ 
\midrule
\multirow{4}{*}{Isotropic}  & $p_1$ & $98/93.1/108.4$ & $97.2/92.2/107$ & $98.1/92/106.2$ & $98.4/91.9/106.3$ & $122.5/133.6/1088.4$ & $\bm{117.7}/128.6/845.2$ & $2$ \\ \cmidrule(lr){2-9}
& $p_2$ & $100.4/93/110.8$ & $98.2/92.1/106.4$ & $97.9/92/106.6$ & $98.1/92/106.3$ & $122.5/134.4/1101.6$ & $117.7/128.4/841.4$ & $4$ \\  \cmidrule(lr){2-9}
& $p_3$ & $98.1/93.2/108.1$ & $96.5/92/106.9$ & $97.5/\bm{91.9}/106.8$ & $98.1/91.9/106.7$ & $122.2/135.7/1124.6$ & $117.8/128.4/840.4$ & $6$ \\  \cmidrule(lr){2-9}
& $p_4$ & $101.7/92.8/114.4$ & $98.5/92.2/106.8$ & $97.6/92/106.6$ & $97.9/92.1/107$ & $122.3/135.8/1105.2$ & $117.9/128.4/841.9$ & $8$ \\ \bottomrule 
\end{tabular}
\end{adjustbox}

\vspace{15pt}

\centering
\caption{Computational time used in each scenario, measured in seconds. For further details, see the caption of Figure~\ref{tab:scores}.}
\label{tab:times}
\begin{adjustbox}{width=\textheight}
\begin{tabular}{rccccc | ccl}
\toprule
& \multicolumn{5}{c|}{Vecchia} & \multicolumn{3}{c}{Composite} \\ \cmidrule(lr){2-6} \cmidrule(lr){7-9}
& Ordering & $d=2$ & $d=3$ & $d=4$ & $d=5$ & $d=2$ & $d=3$  & $m$ \\ \midrule
\multirow{4}{*}{Anisotropic}  & $p_1$ & $201/24/120$ & $1454/513/894$ & $2790/2384/2051$ & $6319/6699/9681$ & $148/372/171$ & $1266/1605/2304$ & $2$ \\ \cmidrule(lr){2-9}
& $p_2$ & $329/155/152$ & $796/469/527$ & $1406/1703/2779$ & $4796/6399/6928$ & $573/377/335$ & $2243/2442/3833$ & $4$ \\  \cmidrule(lr){2-9}
& $p_3$ & $284/137/93$ & $1026/459/868$ & $3593/2245/2588$ & $13670/5335/8566$ & $480/548/565$ & $3291/3381/3535$ & $6$ \\  \cmidrule(lr){2-9}
& $p_4$ & $213/190/115$ & $840/464/624$ & $1836/2654/2330$ & $9902/5032/7156$ & $688/878/722$ & $4972/4823/6493$ & $8$ \\  \midrule
\multirow{4}{*}{Isotropic}  & $p_1$ & $37/13/18$ & $104/56/176$ & $372/290/470$ & $1136/829/1258$ & $21/31/33$ & $137/262/253$ & $2$ \\ \cmidrule(lr){2-9}
& $p_2$ & $20/16/19$ & $97/70/92$ & $350/242/379$ & $998/954/1056$ & $39/56/58$ & $417/452/462$ & $4$ \\  \cmidrule(lr){2-9}
& $p_3$ & $19/13/30$ & $97/73/84$ & $340/273/414$ & $1134/857/1442$ & $55/81/90$ & $501/613/623$ & $6$ \\  \cmidrule(lr){2-9}
& $p_4$ & $19/12/18$ & $101/62/84$ & $370/218/332$ & $1042/722/1179$ & $69/100/108$ & $560/937/843$ & $8$ \\ \bottomrule 
\end{tabular}
\end{adjustbox}
\end{sidewaystable}
Strikingly, the Vecchia likelihood approximation is uniformly better than its composite likelihood counterpart, except in two cases ($d=2$, $\alpha=1.5$, using random or maximum-minimum ordering), which give slightly worse results than the best composite likelihood estimator all settings combined. Overall, the Vecchia likelihood estimator thus clearly outperforms the composite likelihood estimator by a large margin, whatever the ordering (for Vecchia estimators) and cutoff distance (for composite estimators). It is also interesting to note that composite methods in the isotropic case do not even find that $\alpha=1$ is better than $\alpha=0.5$, while all other cases give strong support for $\alpha=1$. Moreover, composite methods perform very poorly when $\alpha=1.5$, while the fits are much more reasonable for Vecchia methods, suggesting that composite methods are less reliable. In terms of computational time, reported in Table~\ref{tab:times}, the Vecchia likelihood estimator is also often much faster than the composite likelihood estimator for fixed $d$. In particular, the Vecchia likelihood estimator with $d=2$ 
only takes a few minutes to run, and already outperforms the best traditional composite likelihood estimator in terms of its log score, even when $d=3$. Moreover, the increase in computational cost as the cutoff dimension $d$ increases is often very large for traditional composite likelihoods, but relatively moderate for the Vecchia likelihood approach. Hence, the Vecchia likelihood estimator is both statistically and computationally more efficient, and easy to implement, which provides strong support for using it in practice.

Overall, our proposed inference approach based on the Vecchia approximation thus delivers excellent results. From Table~\ref{tab:scores}, it is evident that the best results are obtained for $\alpha=1$, with moderate but visible improvements in the anisotropic case. In the best case (anisotropic model with $\alpha=1$, fitted using the Vecchia estimator with middle-out ordering and four conditioning sites, i.e., $d=5$), the parameter estimates are 
$\widehat\lambda=113.69$~km, $\widehat a=0.73$ and $\widehat\theta=0.40$, with $95\%$ confidence intervals $\lambda\in(98.49,128.25)$, $a\in(0.65,0.79)$, and $\theta \in (0.12,0.50)$, obtained from a parametric bootstrap with $300$ bootstrap replicates. These confidence intervals are very similar to those obtained from the (computationally cheaper) jackknife method: $\lambda\in(96.85,130.54)$, $a\in(0.66,0.80)$, and $\theta\in (0.29,0.51)$. The confidence intervals for $a$ clearly exclude the value 1, suggesting the data are indeed anisotropic. Figure~\ref{fig:variogram_contour} displays the contours of the fitted bivariate extremal coefficient $\theta(\boldsymbol{h})=1.1,\ldots,1.9$ with respect to the location at the center of the Red Sea, as described in \eqref{eq:extremalcoef}, based on the best model. The elliptical shape of the estimated contours is well aligned with the geometry of the Red Sea, with stronger dependence along its main axis, which is physically meaningful.

To further compare the Vecchia and composite likelihood approaches, we study the goodness-of-fit of the best-fitting models in each case by comparing the estimated bivariate extremal coefficients, binned across distance classes, with their empirical counterparts. Figure~\ref{fig:Ext_compare} shows the estimated extremal coefficients, plotted as a function of the Mahalanobis distance $d(\bm s_1,\bm s_2) = \sqrt{(\bm s_1-\bm s_2)^T\widehat A(\bm s_1-\bm s_2)}$ where $\widehat A$ is the estimated rotation matrix, for the best isotropic and anisotropic models obtained using the Vecchia and composite approaches. Notice that in the isotropic case, $\widehat A$ is simply the identity matrix, so that $d(\bm s_1,\bm s_2)$ is the Euclidean distance. While the models fitted using the Vecchia method capture the spatial extremal dependence very well at all distances, the fits are poor when using the composite likelihood approach, especially at long distances in the isotropic case. This strongly reinforces the benefits of using the Vecchia likelihood estimator.

\begin{figure}[!t]
	\centering
	\includegraphics[width=0.4\textwidth]{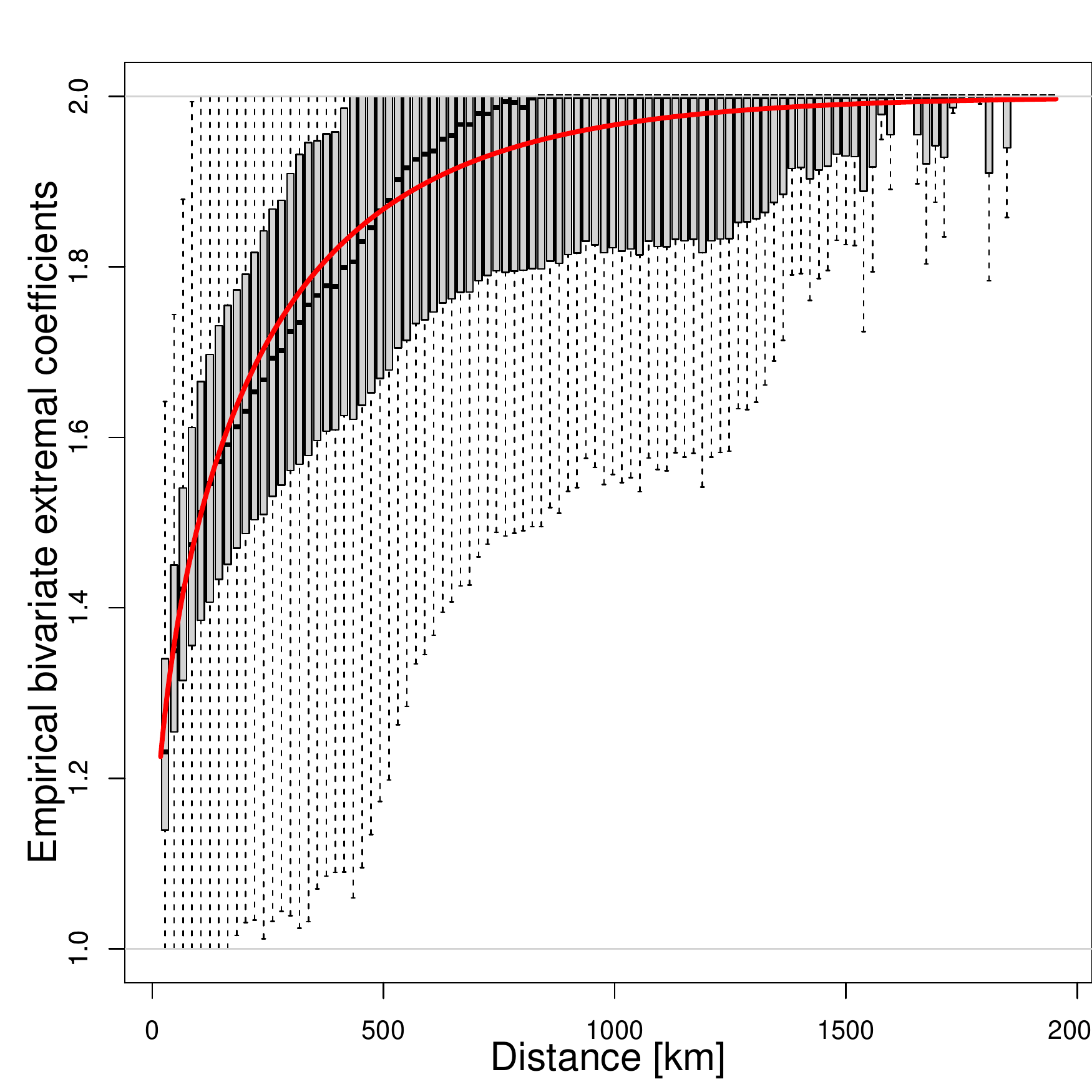}
	\includegraphics[width=0.4\textwidth]{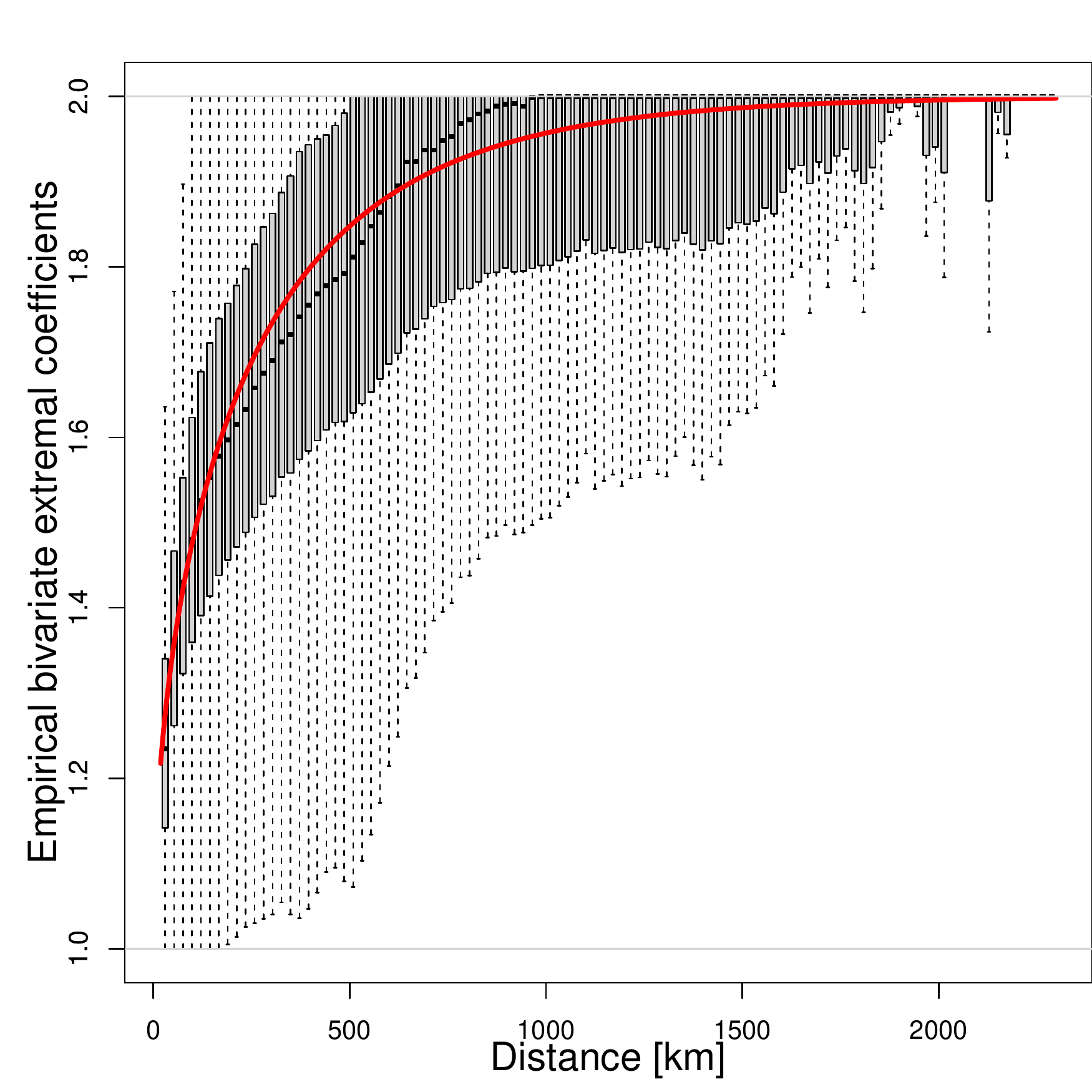}
	\includegraphics[width=0.4\textwidth]{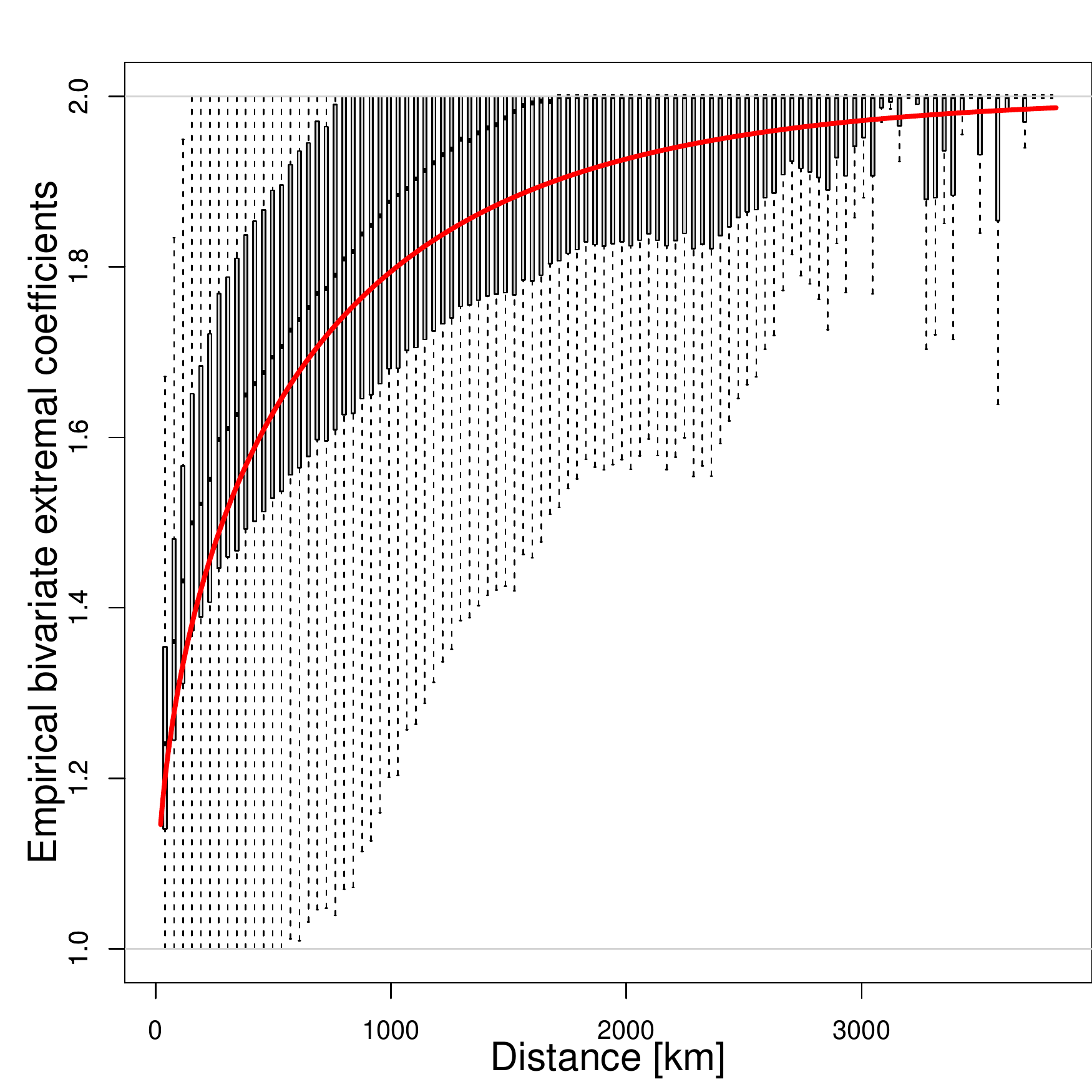}
	\includegraphics[width=0.4\textwidth]{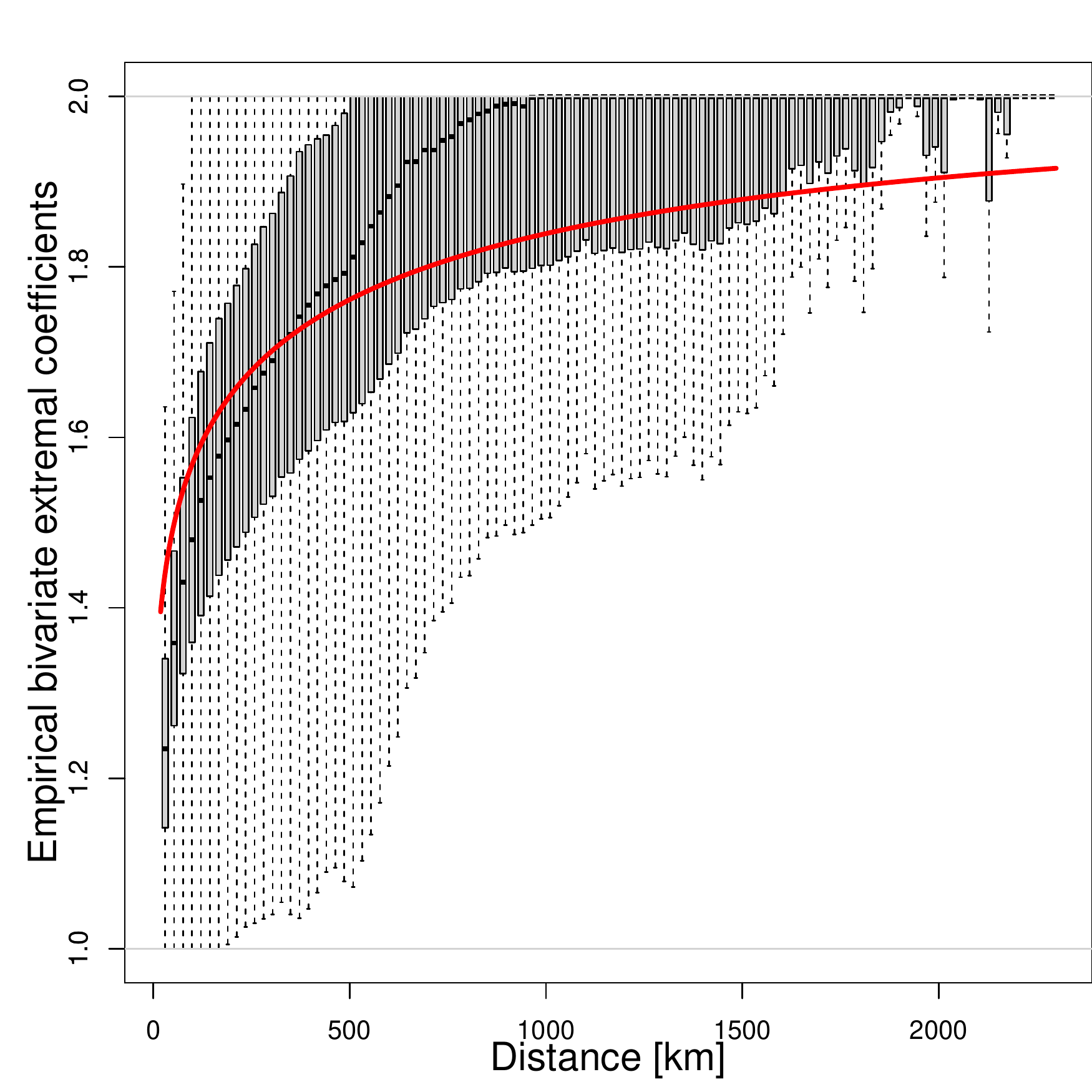}
	\caption{Binned bivariate empirical extremal coefficients (black boxplots), plotted as a function of the Mahalanobis distance $d(\bm s_1,\bm s_2)$, and their model-based counterparts (red curves) for the best anisotropic models (left) and isotropic models (right) obtained using the Vecchia likelihood approach (top) and traditional composite likelihood approach (bottom). The settings of these four ``best models'' can be read from Table~\ref{tab:scores}.}\label{fig:Ext_compare}
\end{figure}

We then also compare the performance of the best anisotropic models (for both Vecchia and composite likelihood approaches) by comparing empirical and fitted extremal coefficients along different directions, and for sub-datasets of different sizes. Specifically, in order to verify the stability of the fitted models, we fit them again using (approximately) $50\%$, $25\%$, $12.5\%$ and $6.25\%$ spread-out sites, chosen according to the maximum-minimum ordering, among the $1043$ sites from the complete dataset. Figure~\ref{fig:Ext_compare} shows plots of the estimated bivariate coefficients for direction-specific pairs of sites in the different sub-panels, plotted against the Euclidean distance between sites. More precisely, binned empirical estimates are compared with model-based estimates for $12$ prevailing directions, namely $15^\circ,30^\circ,\ldots,165^\circ$ (from the East direction in a counterclockwise manner). Figure~\ref{fig:Ext_compare} shows the results for six selected directions, and the Supplementary Material provides results for all $12$ directions. Again, the Vecchia likelihood estimator is able to deliver good and consistent performances in all cases, while the composite likelihood estimator fails completely for some specific directions (see, e.g., the sub-panels corresponding to $15^\circ$ or $75^\circ$). These differences again prove the superiority of the Vecchia method when compared to traditional composite likelihood methods.

\begin{figure}[!t]
	\centering
	 \includegraphics[width=0.9\textwidth]{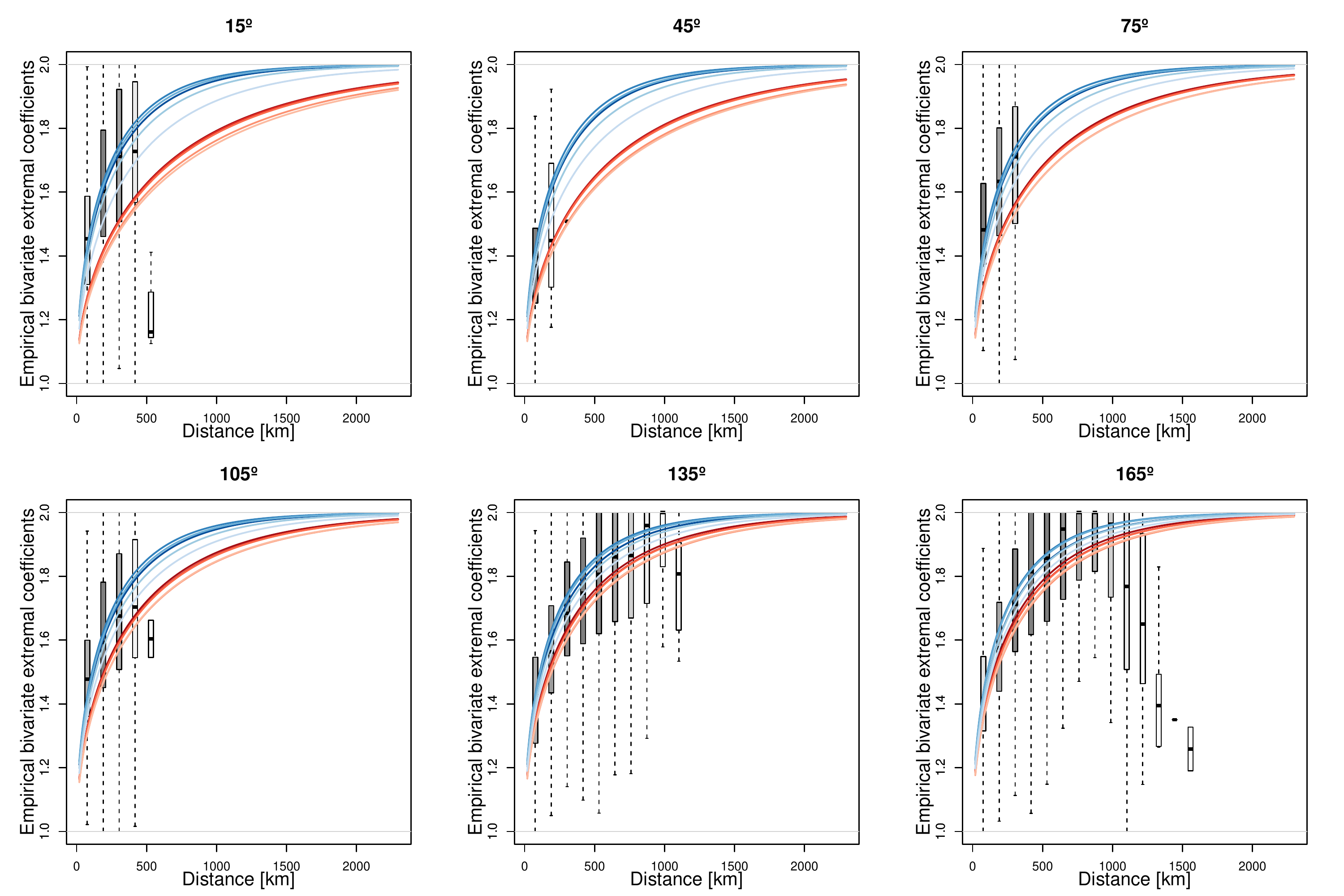}
	\caption{Binned empirical extremal coefficients (boxplots) and their model-based counterparts (colored curves) for different directions, i.e., $15^\circ,45^\circ,\ldots,165^\circ$ (subpanels from top left to bottom right), computed by fitting the best anisotropic models obtained with the Vecchia method (blue curves) and the composite likelihood method (red curves) based on sub-datasets of size $1,1/2, 1/4, 1/8, 1/16\times 100\%$ of the complete dataset (lightest to darkest color). The different shades of grey of the binned boxplots correspond to the number of data points used in each boxplot, with darker grey corresponding to more points.}
	\label{fig:Ext_compare}
\end{figure}

\section{Conclusion}\label{sec:conclusion}
In this paper, we have proposed a new fast and efficient inference method for max-stable processes based on the Vecchia likelihood approximation, which significantly outperforms traditional composite likelihood methods. Unlike pairwise likelihood methods proposed originally by \citet{Padoan.etal:2010} and later extended to higher-order truncated composite likelihoods by \citet{Castruccio.etal:2016} and others, the Vecchia method provides a valid likelihood approximation (i.e., it is itself the likelihood of a well-defined approximated process), thus giving theoretical guarantees to provide improved results, and the number of lower-dimensional likelihood terms involved in it remains linear in the data dimension $D$. Moreover, while it is difficult to choose the cutoff distance $\delta$ and the cutoff dimension $d$ optimally in truncated composite likelihoods, the performance of the Vecchia likelihood estimator is often only moderately sensitive to the choice of the permutation, and always improves as $d$ increases in the Gaussian and max-stable settings we have investigated. Therefore, overall, the Vecchia approximation method is uniformly better than traditional truncated composite likelihoods, be it in terms of statistical efficiency, computational efficiency, ease of implementation, and tuning of parameters. We verified this conclusion in various settings, based on (i) theoretical asymptotic relative efficiency calculations in the case of Gaussian processes, (ii) extensive simulations in the case of max-stable processes, as well as (iii) a substantial real data application to sea surface temperature extremes measured over the whole Red Sea at more than a thousand sites. Our results thus suggest that the superiority of the Vecchia likelihood estimator holds more generally and can be applied in other spatial contexts where the likelihood function is intractable or difficult to evaluate in high dimensions. Finally, while the cutoff dimension $d$ cannot be too big for popular max-stable processes such as the Brown--Resnick model, we have found that the Vecchia approximation method already provides satisfactory results for relatively small $d$, e.g., $d=3$ or $4$, providing a good trade-off between computational and statistical efficiency and major improvements compared to the pairwise likelihood case with $d=2$.

\section*{Acknowledgments}
This publication is based upon work supported by the King Abdullah University of Science and Technology (KAUST) Office of Sponsored Research (OSR) under Awards No. OSR-CRG2017-3434 and No. OSR-CRG2020-4394. Part of the effort of Michael L. Stein is based on work supported by the U.S. Department of Energy, Office of Science, Office of Advanced Scientific Computing Research (ASCR) under Contract DE-AC02-06CH11347. Support from the KAUST Supercomputing Laboratory is also gratefully acknowledged.

\baselineskip=24pt
\appendix
\section*{Appendix}
\section{General expressions for the asymptotic variance in the Gaussian case}\label{subsec:general}
We here derive the asymptotic variance of the composite likelihood estimator $\widehat{\boldsymbol{\psi}}_C$ in \eqref{eq:composite} for Gaussian processes. These general theoretical results are used in Section~\ref{sec:Gaussian} and the Supplementary Material to perform a formal efficiency comparison between the composite likelihood estimator of order $d$, $\widehat{\boldsymbol{\psi}}_{C;d}$, and the Vecchia likelihood estimator, $\widehat{\boldsymbol{\psi}}_{V;d}$, for different correlation models. Our detailed results extend those of \citet{Stein.etal:2004}. 

In order to calculate the asymptotic variance $\boldsymbol{V}=n^{-1}\boldsymbol{J}^{-1}(\boldsymbol{\psi}_0)\boldsymbol{K}(\boldsymbol{\psi}_0)\boldsymbol{J}^{-1}(\boldsymbol{\psi}_0)$, we need to derive the sensitivity matrix $\boldsymbol{J}(\boldsymbol{\psi})=\E\{-{\partial^2\over\partial\boldsymbol{\psi}\partial\boldsymbol{\psi}^\top} \ell_{C}(\boldsymbol{\psi};\boldsymbol{Z})\}$ and the variability matrix $\boldsymbol{K}(\boldsymbol{\psi})=\var\{{\partial\over\partial\boldsymbol{\psi}} \ell_{C}(\boldsymbol{\psi};\boldsymbol{Z})\}$. In case of the full likelihood estimator, we have $\boldsymbol{J}(\boldsymbol{\psi})=\boldsymbol{K}(\boldsymbol{\psi})$, thus the resulting asymptotic variance is $n^{-1}\boldsymbol{J}^{-1}(\boldsymbol{\psi}_0)$, and the expression is obtained by setting $w_{S}=1$ for $S=\{1,\ldots,D\}$ in \eqref{eq:composite} and all other weights to zero. Suppose now that $\boldsymbol{Z}$ has a multivariate normal distribution with zero mean and covariance matrix $\boldsymbol{\Sigma}(\boldsymbol{\psi})$. From \eqref{eq:composite}, and writing $\boldsymbol{\Sigma}_S(\boldsymbol{\psi})$ to denote the covariance matrix of the subvector $\boldsymbol{Z}_S$, it follows that
\begin{align*}
{\partial\over\partial\psi_i}\ell_{C}(\boldsymbol{\psi};\boldsymbol{Z})&=-{1\over2} \sum_{S\in C_D}w_S\left({\partial\over\partial\psi_i}\log |\boldsymbol{\Sigma}_S(\boldsymbol{\psi})|+\boldsymbol{Z}_S^\top{\partial\over\partial\psi_i}\boldsymbol{\Sigma}_S^{-1}(\boldsymbol{\psi})\boldsymbol{Z}_S\right),\\
-{\partial^2\over \partial\psi_i\partial\psi_j}\ell_{C}(\boldsymbol{\psi};\boldsymbol{Z})&= {1\over 2}\sum_{S\in C_D}w_S\left({\partial^2\over \partial\psi_i\psi_j}\log|\boldsymbol{\Sigma}_S(\boldsymbol{\psi})|+\boldsymbol{Z}_S^\top{\partial^2\over \partial\psi_i\psi_j}\boldsymbol{\Sigma}_S^{-1}(\boldsymbol{\psi})\boldsymbol{Z}_S\right),
\end{align*}
for all $i,j=1,\ldots,m$. Now, because the trace is a linear and cyclic operator, we have that 
\begin{align*}
\E\left\{\boldsymbol{Z}_S^\top{\partial^2\over \partial\psi_i\psi_j}\boldsymbol{\Sigma}_S^{-1}(\boldsymbol{\psi})\boldsymbol{Z}_S\right\}&={\rm tr}\left[\E\left\{\boldsymbol{Z}_S^\top{\partial^2\over \partial\psi_i\psi_j}\boldsymbol{\Sigma}_S^{-1}(\boldsymbol{\psi})\boldsymbol{Z}_S\right\}\right]=\E\left[{\rm tr}\left\{\boldsymbol{Z}_S^\top{\partial^2\over \partial\psi_i\psi_j}\boldsymbol{\Sigma}_S^{-1}(\boldsymbol{\psi})\boldsymbol{Z}_S\right\}\right]\\
&=\E\left[{\rm tr}\left\{\boldsymbol{Z}_S\boldsymbol{Z}_S^\top{\partial^2\over \partial\psi_i\psi_j}\boldsymbol{\Sigma}_S^{-1}(\boldsymbol{\psi})\right\}\right]={\rm tr}\left\{\boldsymbol{\Sigma}_S(\boldsymbol{\psi}){\partial^2\over \partial\psi_i\psi_j}\boldsymbol{\Sigma}_S^{-1}(\boldsymbol{\psi})\right\}.
\end{align*}
This implies that the $(i,j)$th entry of the sensitivity matrix is
\begin{equation}\label{eq:J}
\boldsymbol{J}_{i,j}(\boldsymbol{\psi})={1\over 2}\sum_{S\in C_D}w_S\left[{\partial^2\over \partial\psi_i\psi_j}\log|\boldsymbol{\Sigma}_S(\boldsymbol{\psi})|+{\rm tr}\left\{\boldsymbol{\Sigma}_S(\boldsymbol{\psi}){\partial^2\over \partial\psi_i\psi_j}\boldsymbol{\Sigma}_S^{-1}(\boldsymbol{\psi})\right\}\right].
\end{equation}
Moreover, thanks to the Gaussianity assumption, we have that
$$\cov\left\{\boldsymbol{Z}_{S_1}^\top{\partial\over\partial\psi_i}\boldsymbol{\Sigma}_{S_1}^{-1}(\boldsymbol{\psi})\boldsymbol{Z}_{S_1},\boldsymbol{Z}_{S_2}^\top{\partial\over\partial\psi_j}\boldsymbol{\Sigma}_{S_2}^{-1}(\boldsymbol{\psi})\boldsymbol{Z}_{S_2}\right\}=2\;{\rm tr}\left\{{\partial\over\partial\psi_i}\boldsymbol{\Sigma}_{S_1}^{-1}(\boldsymbol{\psi})\boldsymbol{\Sigma}_{S_1,S_2}{\partial\over\partial\psi_j}\boldsymbol{\Sigma}_{S_2}^{-1}(\boldsymbol{\psi})\boldsymbol{\Sigma}_{S_2,S_1}\right\},$$
where $\boldsymbol{\Sigma}_{S_1,S_2}$ is the covariance matrix between the random subvectors $\boldsymbol{Z}_{S_1}$ and $\boldsymbol{Z}_{S_2}$, and $\boldsymbol{\Sigma}_{S_2,S_1}=\boldsymbol{\Sigma}_{S_1,S_2}^\top$. Therefore, the $(i,j)$th entry of the variability matrix is
\begin{equation}\label{eq:K}
\boldsymbol{K}_{i,j}(\boldsymbol{\psi})={1\over 2}\sum_{S_1\in C_D}\sum_{S_2\in C_D}w_{S_1}w_{S_2}\;{\rm tr}\left\{{\partial\over\partial\psi_i}\boldsymbol{\Sigma}_{S_1}^{-1}(\boldsymbol{\psi})\boldsymbol{\Sigma}_{S_1,S_2}{\partial\over\partial\psi_j}\boldsymbol{\Sigma}_{S_2}^{-1}(\boldsymbol{\psi})\boldsymbol{\Sigma}_{S_2,S_1}\right\}.
\end{equation}
Expressions \eqref{eq:J} and \eqref{eq:K} involve derivatives of the log determinant and the inverse covariance matrix, which may be conveniently expressed for all $i,j=1,\ldots,m$ as
\begin{align*}
{\partial\over\partial\psi_i}\log|\boldsymbol{\Sigma}_{S}(\boldsymbol{\psi})|&={\rm tr}\left\{\boldsymbol{\Sigma}_{S}^{-1}(\boldsymbol{\psi}){\partial\over\partial\psi_i}\boldsymbol{\Sigma}_{S}(\boldsymbol{\psi})\right\};\\
{\partial^2\over\partial\psi_i\partial\psi_j}\log|\boldsymbol{\Sigma}_{S}(\boldsymbol{\psi})|&={\rm tr}\left\{-\boldsymbol{\Sigma}_{S}^{-1}(\boldsymbol{\psi}){\partial\over\partial\psi_i}\boldsymbol{\Sigma}_{S}(\boldsymbol{\psi})\boldsymbol{\Sigma}_{S}^{-1}(\boldsymbol{\psi}){\partial\over\partial\psi_j}\boldsymbol{\Sigma}_{S}(\boldsymbol{\psi}) + \boldsymbol{\Sigma}_{S}^{-1}(\boldsymbol{\psi}){\partial^2\over\partial\psi_i\partial\psi_j}\boldsymbol{\Sigma}_{S}(\boldsymbol{\psi})\right\};\\
{\partial\over\partial\psi_i}\boldsymbol{\Sigma}_{S}^{-1}(\boldsymbol{\psi})&=-\boldsymbol{\Sigma}_{S}^{-1}(\boldsymbol{\psi}){\partial\over\partial\psi_i}\boldsymbol{\Sigma}_{S}(\boldsymbol{\psi})\boldsymbol{\Sigma}_{S}^{-1}(\boldsymbol{\psi});\\
{\partial^2\over \partial\psi_i\partial\psi_j}\boldsymbol{\Sigma}_{S}^{-1}(\boldsymbol{\psi})&=\boldsymbol{\Sigma}_{S}^{-1}(\boldsymbol{\psi}){\partial\over\partial\psi_i}\boldsymbol{\Sigma}_{S}(\boldsymbol{\psi})\boldsymbol{\Sigma}_{S}^{-1}(\boldsymbol{\psi}){\partial\over\partial\psi_j}\boldsymbol{\Sigma}_{S}(\boldsymbol{\psi})\boldsymbol{\Sigma}_{S}^{-1}(\boldsymbol{\psi}) \\
& + \boldsymbol{\Sigma}_{S}^{-1}(\boldsymbol{\psi}){\partial\over\partial\psi_j}\boldsymbol{\Sigma}_{S}(\boldsymbol{\psi})\boldsymbol{\Sigma}_{S}^{-1}(\boldsymbol{\psi}){\partial\over\partial\psi_i}\boldsymbol{\Sigma}_{S}(\boldsymbol{\psi})\boldsymbol{\Sigma}_{S}^{-1}(\boldsymbol{\psi}) \\
& - \boldsymbol{\Sigma}_{S}^{-1}(\boldsymbol{\psi}){\partial^2\over \partial\psi_i\psi_j}\boldsymbol{\Sigma}_{S}(\boldsymbol{\psi})\boldsymbol{\Sigma}_{S}^{-1}(\boldsymbol{\psi}).
\end{align*}



\newpage

\baselineskip=16pt

\bibliographystyle{CUP}
\bibliography{Total}



\end{document}